\documentclass[a4paper,reqno,11pt]{article}

\usepackage[hmargin=2.2cm,vmargin=2.2cm]{geometry}

\usepackage{amsmath,amssymb,amsthm,mathtools,cancel,mathrsfs,mathdots,color,graphicx,framed,enumitem,enumerate,todonotes,caption,subcaption}
\usepackage{comment}
\usepackage[colorlinks=true, pdfstartview=FitV, urlcolor=blue, citecolor=red, linkcolor=blue]{hyperref}

\usepackage[utf8]{inputenc}
\usepackage[T1]{fontenc}

\def\Im{\mathrm {Im}\,}
\def\C{\mathbb{C}}
\def\R{\mathbb{R}}

\def\1{\mathbf{1}}
\def\d{\mathrm d}
\def\e{\mathrm{e}}
\def\i{\mathrm{i}}
\def\pa{\partial}

\def\Ai{\mathrm{Ai}}
\def\s{\sigma}

\def\be{\begin{equation}}
	\def\ee{\end{equation}}

\newtheorem{theorem}{Theorem}[section]

\newtheorem{remark}[theorem]{Remark}

\newtheorem{proposition}[theorem]{Proposition} 
\newtheorem{corollary}[theorem]{Corollary}

\usepackage{tikz}
\usetikzlibrary{arrows}
\usetikzlibrary{decorations.pathmorphing}
\usetikzlibrary{decorations.markings}
\usetikzlibrary{patterns}
\usetikzlibrary{automata}
\usetikzlibrary{positioning}
\usepackage{tikz-cd}
\tikzset{->-/.style={decoration={
			markings,
			mark=at position #1 with {\arrow{latex}}},postaction={decorate}}}

\tikzset{-<-/.style={decoration={
			markings,
			mark=at position #1 with {\arrowreversed{latex}}},postaction={decorate}}}

\usetikzlibrary{shapes.misc}\tikzset{cross/.style={cross out, draw, 
		minimum size=2*(#1-\pgflinewidth), 
		inner sep=0pt, outer sep=0pt}}

\begin{document}
	
	\numberwithin{equation}{section}
	
	\title{On the integrable structure of deformed sine kernel determinants}
	\author{Tom Claeys$^1$ and Sofia Tarricone$^2$}
	\date{}
	\maketitle
	
	\begin{center}
		{$^1$\footnotesize
			{\textit{Institut de Recherche en Math\'ematique et Physique, Universit\'e catholique de Louvain,
					\\
					Chemin du Cyclotron 2, 1348 Louvain-la-Neuve, Belgium}}
				\\
				\texttt{tom.claeys@uclouvain.be}
			\\
			$^2$\footnotesize
			{\textit{Institut de Physique Th\'eorique, Universit\'e Paris-Saclay, CEA, CNRS,
					\\F-91191 Gif-sur-Yvette, France}
				\\
				\texttt{sofia.tarricone@ipht.fr}}
		}
	\end{center}

\medskip
	
	\begin{abstract} We study a family of Fredholm determinants associated to deformations of the sine kernel, parametrized by a weight function $w$. For a specific choice of $w$, this kernel describes bulk statistics of finite temperature free fermions. We establish a connection between these determinants and a system of integro-differential equations generalizing the fifth Painlev\'e equation, and we show that they allow us to solve an integrable PDE explicitly for a large class of initial data.
	\end{abstract}




	
	\section{Introduction}
\paragraph{Context.}	
On microscopic scales, eigenvalues of large classes of random matrices are distributed according to a limited number of universal determinantal point processes. The most classical examples of such point processes are the sine process, which arises in bulk scaling limits, and the Airy point process, which corresponds to (soft) edge scaling limits.
They are determinantal point processes on the real line with correlation kernels
\[\mathrm K^{\sin}(x,y)=\frac{\sin(\pi(x-y))}{\pi(x-y)},\qquad
\mathrm K^{\Ai}(x,y)=\int_{0}^{+\infty}\Ai(x+u)\Ai(y+u)\d u,\]
where $\Ai$ is the Airy function.
In recent years, particular deformations of the sine and Airy kernels, along with the associated deformed determinantal point processes, have attracted interest because of their connection with models of free fermions at finite temperature \cite{DeanLedoussalMajumdarSchehr,  LeDoussalMajumdarRossoSchehr} and with the Moshe-Neuberger-Shapiro model \cite{Johansson, LiechtyWang, MNS}. 
A Fredholm determinant of this finite-temperature deformation of the Airy kernel moreover characterizes the narrow wedge solution of the Kardar-Parisi-Zhang equation \cite{AmirCorwinQuastel}, and has been studied intensively because of this connection. More recently, another deformation of the Airy and sine kernel determinants appeared in the study of limiting behavior for edge and bulk  spacing distributions of complex elliptic Ginibre matrices in weak non-hermiticity limits \cite{BothnerLittle, BothnerLittlebulk}.

In this work, our aim is to study a large class of Fredholm determinants of finite-temperature type deformations of the sine kernel. We will show that, just like in the Airy case, these determinants have a rich integrable structure: they are intimately connected to an integro-differential Painlev\'e equation, and they are the key objects for the direct and inverse scattering transform of a Zakharov-Shabat system, which allows us to solve an integrable PDE introduced by Its, Izergin, Korepin, and Slavnov \cite{IIKS} in a remarkably simple and explicit manner in terms of its initial data.

\paragraph{The sine kernel determinant.}	

Let $\mathcal K^{\sin}$ be the integral operator  with kernel $\mathrm K^{\sin}$,
\[\mathcal K^{\sin} f(x)=\int_{-\infty}^{+\infty}\mathrm K^{\sin}(x,y)f(y)\d y.\]
We define the sine kernel Fredholm determinant as follows,
\be
F(s;\ell):=\det\left(1-\ell \mathcal K^{\sin}|_{\left[-\frac{s}{2\pi},\frac{s}{2\pi}\right]}\right):=1+\sum_{n=1}^\infty\frac{(-\ell)^n}{n!}
\int_{\left[-\frac{s}{2\pi},\frac{s}{2\pi}\right]^n}\det\left(\mathrm K^{\sin}(x_j,x_k)\right)_{j,k=1}^n\prod_{j=1}^n\d x_j,
\ee
for $s>0,\ell\in\mathbb C$,
where we recognize the series at the right as the standard Fredholm series. For $\ell=1$, this quantity $F(s;1)$ is equal to the probability that a random point configuration in the sine process contains no points in the interval $\left[-\frac{s}{2\pi},\frac{s}{2\pi}\right]$; for $0<\ell<1$, $F(s;\ell)$ is the probability that a random point configuration in the thinned sine process, obtained by removing each point in a configuration independently with probability $1-\ell$, contains no points in the interval $\left[-\frac{s}{2\pi},\frac{s}{2\pi}\right]$.

Jimbo, Miwa, Mori, and Sato \cite{JMMS} first and Tracy and Widom \cite{TW94} later proved that $F(s;\ell)$ is expressible in terms of a Painlev\'e transcendent: they proved 
	\be 
	\label{eq:jmms formula sine}
	F(s;\ell)= \exp\left(\int_0^{s} \frac{\nu(x;\ell)}{x}\d x\right),
	\ee
		where $\nu\coloneqq\nu(x;\ell)$ solves (a special case of) the $\sigma$-form of the fifth Painlev\'e equation, namely the second order ODE (here $'$ denotes the $x$-derivative)
	\be 
	\label{eq:PVsigmaform}
	(x\nu'')^2+4(x\nu'-\nu)(x\nu'-\nu + (\nu')^2)=0,
	\ee
	together with the boundary condition 
	\be
	\nu(x;\ell)=-\frac{\ell}{\pi}x+O(x^2), \qquad x\rightarrow 0.
	\ee
	For later convenience, we remark that this implies in particular that
		\be
		\label{eq:nu and fred det zero t}
		\pa_s \log F(s;\ell) =  \frac{\nu(s;\ell)}{s},
		\ee
		and so
		\be
		\label{eq: second log der Q zero t}
		\pa_s s \pa_s \log F(s;\ell) = \nu'(s;\ell).  
		\ee

\paragraph{Deformed sine kernel determinants.}
Observe that we can express the sine kernel as
\be
	\mathrm K^{\sin} (x,y) = \int_{0}^{1} \cos \pi (x-y) t \d t = \int_{-1/2}^{1/2} \e^{ 2\pi \i(x-y) u} \d u.
	\ee
Using this last expression, we can factorize the associated integral operator acting on $L^2(\mathbb R)$ as follows,
	\be
	\label{eq:sinekernel decomp0}
	\mathcal{K}^{\sin} = \mathcal{F}^* \chi_{(-1/2,1/2)}\mathcal{F}=\mathcal{F} \chi_{(-1/2,1/2)}\mathcal{F}^*,
	\ee
	where $\chi_I$ denotes the projection operator on the set $I$, and $\mathcal{F}$ and $\mathcal{F}^*$ denote the Fourier transform and its inverse acting on $L^2(\mathbb R)$, 
	\be
	\mathcal{F}f(x)= \int_\R f(t)\e^{-2\pi\i xt}\d t, \;\;\text{and}\;\; \mathcal{F}^*g(y)= \int_\R g(u)\e^{2\pi\i yu}\d u.
	\ee
We can deform the sine kernel by replacing the projection operator in \eqref{eq:sinekernel decomp0} by a multiplication operator, 	\be
	\label{eq:sinekernel decomp deformed}
	\mathcal{K}_w^{\sin} = \mathcal{F}^* \mathcal M_w \mathcal{F},\qquad \mathcal M_w f(u)=w(u)f(u),
	\ee
	for any integrable function $w:\mathbb R\to \mathbb C$, such that the associated kernel is given by
\be\label{eq:deformedsinekernel} \mathrm K_w^{\sin}(x,y)=\int_{-\infty}^{\infty} \e^{ 2\pi \i(x-y) u} w(u) \d u.
\ee
We consider the Fredholm determinant
\be
\label{eq:def Fw}
F_w(s):=\det\left(1-\mathcal K_{w}^{\sin}|_{\left[-\frac{s}{2\pi},\frac{s}{2\pi}\right]}\right).
\ee
As we will show later, $\mathcal K_{w}^{\sin}|_{\left[-\frac{s}{2\pi},\frac{s}{2\pi}\right]}$ is trace-class (see Proposition \ref{prop:traceclass}), such that the Fredholm determinant is well-defined, and we have the alternative Fredholm determinant identity
\be
\label{eq:def Fw-alt}
F_w(s):=\det\left(1-\sqrt{w_s}\mathcal K^{\sin}\sqrt{w_s}\right),
\ee
where $\mathcal K^{\sin}$ is the non-deformed sine kernel operator, and $\sqrt{w_s}$ denotes the multiplication operator with a square root (note that the operator $\sqrt{w_s}\mathcal K^{\sin}\sqrt{w_s}$ is independent of the choice of square root) of the function $w_s(r):=w(\frac{\pi r}{s})$ (see Proposition \ref{prop: fred det equalities}).
The identity \eqref{eq:def Fw-alt} has important advantages, since it allows us to characterize $F_w(s)$ in terms of a $2\times 2$ matrix Riemann-Hilbert (RH) problem (see Section \ref{section:RH}).

If $w$ moreover takes values in $[0,1]$, the operator $\mathcal K_w^{\sin}$ is Hermitian and we will show (see also Proposition \ref{prop:traceclass}) that $0\leq \mathcal K_w^{\sin}< 1$. It then follows from \cite{Soshnikov} that $\mathcal K_w^{\sin}$ defines a unique determinantal point process, which we call the ($w$-)deformed sine process.
The Fredholm determinant \eqref{eq:def Fw} is then equal to the probability that a random configuration in this deformed sine process has no points in $\left[-\frac{s}{2\pi},\frac{s}{2\pi}\right]$.

For general integrable $w:\mathbb R\to\mathbb C$, the second Fredholm determinant identity \eqref{eq:def Fw-alt} also has a natural point process interpretation: it is the average multiplicative statistic $\mathbb E\prod_{j=1}^\infty (1-w_s(x_j))$, where the average is with respect to the (non-deformed) sine process, and $x_1,x_2,\ldots$ denote the random points in a configuration.
\begin{remark}\label{remark:FT}
The finite temperature sine kernel from \cite{DeanLedoussalMajumdarSchehr, Johansson, LiechtyWang} is given by
\be 
{K}_\alpha^{\rm FTsin}(x,y) = \int_{0}^{+\infty}\frac{\cos \pi (x-y)u}{\alpha\e^{u^2}+1}\d u
=\int_{-\infty}^{+\infty}\frac{\e^{2\pi\i(x-y)u}}{\alpha\e^{4u^2}+1}\d u
,
\ee
where $\alpha \in \left(0,1\right)$, and $-\frac{1}{\log\alpha}$ can be interpreted physically, in a free fermion model, as a parameter proportional to the temperature of the system. 
This is precisely our deformed sine kernel $\mathrm K_{w}^{\sin}$ in \eqref{eq:deformedsinekernel} for the special choice of $w$,
\be
w(u)=\frac{1}{\alpha\e^{4u^2}+1}.
\ee
Observe that this is an even function taking values in $(0,1)$.
\end{remark}
\begin{remark} Recently, \cite{GLDS1, GLDS2} studied other models of free fermions, different from the ones in e.g.\ \cite{DeanLedoussalMajumdarSchehr}. These models are concerned with potentials with delta impurity. Already at zero temperature, certain limiting behaviors of the fermions in the bulk are again described by a deformed sine kernel of the type \eqref{eq:deformedsinekernel} but with $w(u)=h(u)\chi_I(u)$, $h$ being a smooth function (depending on the potential) and $I$ a compact real interval. As we will see, this choice of $w$ does not actually belong to the class for which our results hold, and it would be interesting to see how our results generalize to such situations.
\end{remark}

\paragraph{Objectives.}
As already mentioned, the deformed sine kernel determinants introduced above are natural counterparts of deformed Airy kernel determinants which have been investigated intensively in recent years. 
Define the $\sigma$-deformed Airy kernel as \[\mathrm K^{\Ai}_\sigma(x,y)=\int_{\mathbb R}\Ai(x+t)\Ai(y+t)\sigma(t)\d t.\]
For a suitable class of functions $\sigma$, it was shown first in \cite{AmirCorwinQuastel} (see also \cite{Bothner, BothnerCafassoTarricone, CafassoClaeysRuzza, Krajenbrink}) that
\begin{equation}\label{eq:airy fred det}
	\partial_s^2\log\det\left(1-\mathcal K^{\Ai}_\sigma|_{(s,+\infty)}\right)=-\int_{\mathbb R}\varphi_\sigma(\lambda;s)^2\sigma'(\lambda)\d \lambda,
\end{equation}
where $\varphi_\s$ solves the integro-differential Painlev\'e II equation 
\be
\label{IDPII}
\pa_s^2\varphi_\s(\lambda;s)=\biggl(\lambda+s+2\int_\R\varphi_\s(\mu;s)^2\s'(\mu)\d\mu\biggr)\varphi_\s(\lambda;s).
\ee
If $\sigma=\chi_{(0,+\infty)}$, the kernel $\mathrm K^{\Ai}_\sigma$ is the classical Airy kernel, and in that case the integro-differential Painlev\'e II equation degenerates to the Painlev\'e II equation $q''(s)=sq(s)+2q(s)^3$ for $q(s)=\varphi_\sigma(0;s)$, thus recovering the classical Tracy-Widom formula \cite{TracyWidomAiry}.

 Moreover, for a one-parameter family of functions $\sigma=\sigma_t$ of the form $\sigma_t(u)=\sigma(t^{-2/3}u)$,
the quantity $\partial_x^2\log\det\left(1-\mathcal K^{\Ai}_{\sigma_t}|_{(-xt^{-1/3},+\infty)}\right)+\frac{x}{2t}$ solves the KdV equation, see \cite[Theorem 1.3]{CafassoClaeysRuzza}.
 
\medskip

Our objective in this work is to describe similar connections between the deformed sine kernel determinants, an integro-differential Painlev\'e equation, and an integrable PDE.
To avoid technical complications, we will restrict ourselves to Schwartz functions $w:\mathbb R\to\mathbb C$, even though we believe most of the results hold under weaker smoothness and decay assumptions on $w$.
First, we will reveal a simple connection between $F_w(s)$ and the Zakharov-Shabat system, which we can re-write here as a system of integro-differential equations. For $w$ even, it reduces to a single integro-differential equation. We will also explain how we can see this integro-differential equation as a generalization of the fifth Painlev\'e equation. 
Secondly, we will consider a one-parameter family of even functions $w$, namely we will take $w(u)$ of the form $W(u^2-y)$, $y\in\mathbb R$, for some function $W$. Then, the Fredholm determinants $F_{W(.^2-y)}(s)$ are connected to an integrable PDE introduced in \cite{IIKS}. As we will see, they drive the direct and inverse scattering transform which allows to solve this integrable PDE explicitly in terms of its initial data. It is striking that the map transforming initial data into scattering data and vice versa consists only of applying simple integral transforms.

\section{Statement of results}
\paragraph{Connection with the Zakharov-Shabat system.}
Our first result connects the determinant $F_w(s)$ to a system of integro-differential equations, arising from the Zakharov-Shabat system, in general, and to an integro-differential generalization of the Painlev\'e V equation, in the case where $w$ is symmetric. In particular, we are going to prove the following characterization. 
\begin{theorem}\label{thm:s}
Let $w:\mathbb R\to\mathbb C$ be a Schwartz function.
\begin{itemize}
\item[(i)]For every $s>0$ such that $F_w(s)\neq 0$, we have the identity 
	\be
\pa_s s \pa_s \log F_w(s) = \frac{1}{\pi}\int_{\R} \lambda w'(\lambda)\phi(\lambda;s)\psi(\lambda;s)\d \lambda,
\ee
	where $\phi,\psi$ solve the system of equations \begin{align}
	& \pa_s\phi(\lambda;s) = \i \lambda \phi(\lambda,s) -\frac{1}{2\pi\i s} \int_{\R}\phi^2(\mu;s)w'(\mu)\d \mu\ \psi(\lambda;s),\l\label{eq:system phi psi 1}\\
	& \pa_s\psi(\lambda;s) = \frac{1}{2\pi\i s} \int_{\R}\psi^2(\mu;s)w'(\mu)\d \mu\ \phi(\lambda;s) -\i\lambda \psi(\lambda;s).\label{eq:system phi psi 2}
\end{align}
with  {$\lambda\to \pm\infty$ asymptotics} \be{\phi(\lambda;s)\sim \e^{\i s\lambda},\qquad \psi(\lambda;s)\sim \e^{-\i s\lambda}}.
\ee
Moreover, if $w:\mathbb R\to[0,1]$ takes values in $[0,1]$, $F_w(s)\neq 0$ for all $s>0$.
\item[(ii)] If $w$ is even, $w(\lambda)=w(-\lambda)$ for all $u\in\mathbb R$, we have
	\be\label{eq:JMMSdeformed}
\pa_s s\pa_s\log F_w(s) = \frac{1}{\pi} \int_\R \lambda w'(\lambda) \phi (\lambda;s)\phi(-\lambda;s) \d \lambda,
\ee
where $\phi$ solves the integro-differential equation  
\be\label{eq:IDPV}
\pa_s \phi(\lambda;s) = \i \lambda \phi(\lambda;s) -\frac{1}{2\pi \i s} \int_{\R}\phi^2(\mu;s)w'(\mu)\d\mu\ \phi(-\lambda;s),
\ee
with {$\lambda\to \pm\infty$ asymptotics $\phi(\lambda;s)\sim \e^{\i s\lambda}$.}
\end{itemize}
\end{theorem}
\begin{remark}
The system \eqref{eq:system phi psi 1}--\eqref{eq:system phi psi 2} is nothing else than the Zakharov-Shabat system \cite{ZakharovShabat}
\begin{align*}
	& \pa_s\phi(\lambda;s) = \i\lambda \phi(\lambda;s) -\i\beta(s) \psi(\lambda;s),\l\\
	& \pa_s\psi(\lambda;s) = \i\gamma(s) \phi(\lambda;s) -\i\lambda \psi(\lambda;s),
\end{align*}
along with the integral identities
\[\beta(s)=\frac{-1}{2\pi s} \int_{\R}\phi^2(\mu;s)w'(\mu)\d \mu,\qquad \gamma(s)=\frac{-1}{2\pi s} \int_{\R}\psi^2(\mu;s)w'(\mu)\d \mu,\]
for the potentials $\beta$ and $\gamma$. 
Such integral expressions are, in the literature of integrable differential equations, often referred to as trace formulas, see e.g.\ \cite{DeiftTrubowitz}.
\end{remark}

\paragraph{Back to Painlev\'e V.} Since $w$ is assumed to be smooth, the choice $w(\lambda)=\chi_{\left(-\frac{1}{2},\frac{1}{2}\right)}(\lambda)$ is not admissible. However, we can approximate $\chi_{\left(-\frac{1}{2},\frac{1}{2}\right)}(\lambda)$ by a sequence of smooth functions $w$, and then in the limit, it is natural to interpret $w'(\lambda) = \delta_{-\frac{1}{2}}(\lambda)- \delta_{\frac{1}{2}}(\lambda)$ as a sum of $\delta$-functions.
By doing this, the results of Theorem \ref{thm:s} (ii) degenerate to the Painlev\'e V expression \eqref{eq:jmms formula sine}, as we explain next.

In this degenerate case, the integro-differential equation \eqref{eq:IDPV} reduces to the system
\be
\pa_s \phi\left(\frac{1}{2};s\right) = \frac{\i}{2}\phi\left(\frac{1}{2};s\right) -\frac{1}{2\pi\i s} \left(\phi^2\left(-\frac{1}{2};s\right)-\phi^2\left(\frac{1}{2};s\right)\right)\phi\left(-\frac{1}{2};s\right)
\ee
and 
\be
\pa_s \phi\left(-\frac{1}{2};s\right) =  -\frac{\i}{2}\phi\left(-\frac{1}{2};s\right) -\frac{1}{2\pi\i s} \left(\phi^2\left(-\frac{1}{2};s\right)-\phi^2\left(\frac{1}{2};s\right)\right)\phi\left(\frac{1}{2};s\right).
\ee
Hence, defining  $v(x), u(x)$ by  
\be 
\label{eq:def vu}
v(x)=\frac{1}{2\pi\i}\phi\left(\frac{1}{2};\frac{x}{2\i}\right)\phi\left(-\frac{1}{2};\frac{x}{2\i}\right),\qquad u(x) =\frac{\phi^2\left(\frac{1}{2};\frac{x}{2\i}\right)}{\phi^2\left(-\frac{1}{2};\frac{x}{2\i}\right)},
\ee
we see after a straightforward computation that $v,u$ solve the coupled system of differential equations
\begin{equation}
	\label{eq:sys u v}
	xv'=v^2(u-\frac{1}{u}),\quad xu'=xu-2v(u-1)^2.
\end{equation}
This system is well-known to be related to the Painlev\'e V equation, see  \cite{FIKN} or \cite[Equations (4.106-7) with $\alpha=0=\beta$]{CIK}: it implies that $u(x)$ solves the Painlevé V equation (again, see \cite{FIKN} or \cite[Equation (1.22)  with $A=B=0, C=1, D=-1/2$]{CIK}
\be 
\label{eq:PVu}
u''=\frac{u}{x}-\frac{u'}{x}- \frac{u(u+1)}{2(u-1)}+(u')^2\frac{3u-1}{2u(u-1)},
\ee
and that
\[
\sigma(x)=\int_0^{x}\sigma'(\xi)\d\xi\quad\mbox{with}\quad
\sigma'(x):=-\frac{1}{2\pi\i}\phi\left(\frac{1}{2};\frac{x}{2\i}\right)\phi\left(-\frac{1}{2};\frac{x}{2\i}\right),\]
solves the $\sigma$-form of the Painlev\'e V equation
\be\label{eq:PVsigma}
(x\sigma'')^2=(\sigma-x\sigma'+2(\sigma')^2)^2-4(\sigma')^4.
\ee
Defining $\nu(s)=\sigma(2\i s)$, such that
\[\nu'(s)=-\frac{1}{\pi}\phi\left(\frac{1}{2};s\right)\phi\left(-\frac{1}{2};s\right),\]
it is now straightforward to check that $\nu(s)$ solves \eqref{eq:PVsigmaform} and
the above expression in the right hand side is indeed the degenerate case of \eqref{eq:JMMSdeformed} as it should be and it corresponds to equation \eqref{eq: second log der Q zero t} ($\ell=1$).
\begin{remark}
Another integro-differential generalization of the Painlev\'e V equation related to $F_w(s)$ appeared in \cite{BothnerLittlebulk}. For smooth, even and exponentially fast decaying to zero at $\pm\infty$ functions $w$, the authors there established an expression for the Fredholm determinant $F_w(s)$ in terms of a solution of an integro-differential equation of Painlevé type \cite[Theorem 6.7]{BothnerLittlebulk}. More precisely, they gave a formula of the form
	\begin{equation}
		\frac{d^2}{ds^2}\log F_w(s) = -\left(\int_0^\infty r(z;s)w'(z)dz\right)^2, 
	\end{equation}
	where $r(z;s)$ solves an integro--differential equation that for $w= \chi_{(-1/2,1/2)}$ degenerates to an ODE for $r(s)=r(1/2;s)$  (see e.g. \cite[Equation (3.6.36)]{AndersonGuionnetZeitouni}) related to the (special case of the) Painlevé V sigma form  \eqref{eq:PVsigmaform} for $\nu(s)$, by the transformation 
	\[\frac{d}{ds}\left(\frac{\nu(s)}{s}\right)=-r^2(s).\]
For the specific choice of $w$ given by
	\begin{equation}
		\label{eq:weight eGinUE}
	w(z) =\Phi(\alpha(z+1))- \Phi(\alpha(z-1)), \;\text{with}\; \Phi(z)=\frac{1}{\sqrt{\pi}}\int_{-\infty}^z \e^{-y^2}dy,\quad\alpha>0,
	\end{equation}
(where $\alpha=s/\sigma>0$ and $\sigma$ is a parameter) the authors  of \cite{BothnerLittlebulk} also proved	that the Fredholm determinant $F_w(s)$ describes the limiting behavior in the weak non-hermiticity limit of the bulk spacing distribution of the real parts of eigenvalues in the complex elliptic Ginibre ensemble. 
\end{remark}

\paragraph{Second deformation and integrable PDE.}
In order to relate the Fredholm determinants $F_w(s)$ to an integrable PDE, we now introduce an additional deformation, by deforming the function $w=w_y$ depending on a parameter $y\in\mathbb R$. A convenient deformation will turn out to be as follows. 
Let $W:\mathbb R\to\mathbb C$ be smooth and decaying fast at $+\infty$, such that the even function
$w(\lambda)=W(\lambda^2-y)$ is a Schwartz function. We consider the Fredholm determinant
\be\label{def:FW}
Q_W(y,s):=\det\left(1-\mathcal K_{w}^{\sin}|_{\left[-\frac{s}{2\pi},\frac{s}{2\pi}\right]}\right)=F_{w}(s).\ee
By a simple change of variables in the integrals of the Fredholm series, we see that alternative expressions for $Q_W$ and $F_w$ are
\be\label{eq:FW2}
Q_W(y,s)
=\det\left(1-\mathcal K_{w_{y,s}}^{\sin}|_{[-1/2,1/2]}\right),\qquad F_w(s)
=\det\left(1-\mathcal K_{w_{0,s}}^{\sin}|_{[-1/2,1/2]}\right)
\ee
where \be\label{def:wst}
w_{y,s}(\zeta)=W\left(\frac{\pi^2\zeta^2}{s^2}-y\right).\ee Denoting $w_{y,s}$ also for the corresponding multiplication operator, we have, analogous to \eqref{eq:def Fw-alt} and by Proposition \ref{prop: fred det equalities} below, the alternative representation
\be
\label{eq:def QW-alt}
Q_W(y,s):=\det\left(1-\sqrt{w_{y,s}}\mathcal K^{\sin}\sqrt{w_{y,s}}\right).
\ee
\begin{remark}
We see from Remark \ref{remark:FT} that
for the special choice \begin{equation}
		\label{eq:w FT} 
W(r)=\frac{1}{\e^{4r}+1},\end{equation}
we can identify $Q_W(y,s)$ with the Fredholm determinant of the finite temperature sine kernel:
\[\det\left(1-\left.\mathrm K_\alpha^{\rm FTsin}\right|_{\left[-\frac{s}{2\pi},\frac{s}{2\pi}\right]}\right)=Q_{W}\left(y=-\frac{\log\alpha}{4}, s\right).\]
\end{remark}

For $y\in\mathbb R, s>0$ such that $Q_W(y,s)\neq 0$, define $\sigma_W(y,s)$ and $q_W(y,s)$ by
\be\label{def:sigmaq}
\sigma_W(y,s):=\log Q_W(y,s),\qquad
q_W(y,s)^2=-\partial_s^2\sigma_W(y,s).
\ee
Then, $q_W$ and $\sigma_W$ are solutions to PDEs which appeared more than thirty years ago in the groundbreaking work \cite{IIKS} by Its, Izergin, Korepin, and Slavnov, but which do not seem to have been investigated much in recent years.

\begin{theorem}\label{thm:q}
Let $W:\mathbb R\to\mathbb C$ be such that $W(.^2-y)$ is a Schwartz function for every $y\in\mathbb R$.
Then, for any $s>0$, $y\in\mathbb R$ such that $Q_W(y,s)\neq 0$,
$q=q_W$ 
solves the PDE
\be\label{PDE}
\partial_s\left(\frac{\partial_s\partial_y q}{2q}\right)=\partial_y(q^2)-1,
\ee
and $\sigma=\sigma_W$ solves the PDE
\be\label{PDEsigma}
(\partial_s^2\partial_y \sigma)^2=4\partial_s^2\sigma\left(-2s\partial_s\partial_y\sigma+2\partial_y\sigma-(\partial_s\partial_y\sigma)^2\right).
\ee
In particular, if $W$ takes values in $[0,1]$, $q_W(y,s)$ and $\sigma_W(y,s)$ are defined for all $s>0$, $y\in\mathbb R$, and solve \eqref{PDE} and \eqref{PDEsigma}.
\end{theorem}
\begin{remark}
Note that the choice of square root for $q_W(y,s)$ in \eqref{def:sigmaq} is unimportant for the above result, as long as the relevant derivatives exist locally. Indeed, the equation for $q$ is invariant under the map $q\mapsto -q$. Similarly, the choice of logarithm for $\sigma_W(y,s)$ is unimportant as long as the relevant derivatives exist locally, because $\sigma$ in \eqref{PDEsigma} appears only under a $y$- or $s$-derivative.
For our asymptotic results below, it is however important to fix our choice of logarithm in \eqref{def:sigmaq}: we therefore define $\sigma_W(y,s)$ as the smooth in $y$ and $s$ logarithm of $Q_W(y,s)$ which is such that $\lim_{s\to 0}\sigma_W(y,s)=0$.

\end{remark}
\begin{remark}
For $W$ given by \eqref{eq:w FT}, the PDEs for $q_W(y,s)$ and $\sigma_W(y,s)$ were already established in \cite[Equations (5.14--15)]{IIKS}.   This determinant was investigated in \cite{IIKS} to describe the zero temperature equal time one-point correlation function in the impenetrable one dimensional Bose gas \cite[Equation (1.10)]{IIKS}. Even if the 
general techniques developed in \cite{IIKS} form a cornerstone of the modern RH technology which we use here, the connection between $Q_W(y,s)$ and the PDE \eqref{PDEsigma} was derived in \cite{IIKS}
using methods that are quite different from ours. We believe therefore that it is interesting and enlightening to revisit this connection using the RH methods available nowadays, as we do. 
\end{remark}

\paragraph{Asymptotics and scattering.}
We constructed a class of solutions $\sigma_W(y,s)$ to the PDE \eqref{PDEsigma} parametrized by a function $W$. It is then natural to ask to what extent these are general solutions of this PDE, and whether we can identify the function $W$ in terms of properties of the solution $\sigma_W$. To that end, we consider the asymptotic behavior of $\sigma_W(y,s)$ as $s\to 
0$. Afterwards, we will see that these asymptotics allow to find a simple relation between the initial data as $s\to 0$ and the function $W$.
\begin{theorem}\label{thm:assigma}
Let $W:\mathbb R\to\mathbb C$ be such that $W(.^2-y)$ is a Schwartz function for every $y\in\mathbb R$.
Then, $\sigma_W(y,s)$ has the following asymptotics:
\be\label{eq:assigma0}\sigma_W(y,s)\sim -\frac{2s}{\pi}\int_{0}^{+\infty}W(\lambda^2-y)\d \lambda,\ee
as $s\to 0$, and also as $y\to -\infty$ with either $s>0$ fixed or $s\to 0$.
\end{theorem}
Finally, we can employ the above result to construct an explicit solution method for the PDE \eqref{PDEsigma}, as stated next.
\begin{theorem}\label{thm:scattering}
Let $f:\mathbb R\to \mathbb C$ be $C^\infty$ and decaying fast at $-\infty$, such that $f(y-.^2)$ is a Schwartz function for all $y\in\mathbb R$. Define
\be\label{def:WF}
W(r)=-2\int_0^{\infty}f'(-u^2-r)\d u.
\ee
Then, for every $s>0$, $y\in\mathbb R$ for which $\sigma_W(y,s)\neq 0$, the function $\sigma_W(y,s)$ solves the PDE \eqref{PDEsigma} with initial data 
\be
\label{eq:initialdata}
\lim_{s\to 0}\frac{1}{s}\sigma_W(y,s)=f(y).\ee
\end{theorem}
\begin{remark}
The expression for $W$ in terms of $f$ and the inverse given in \eqref{eq:assigma0} can alternatively be rewritten in terms of the direct and inverse Abel integral transform, after an explicit change of variable. For the sake of simpler and more explicit formulas, we choose however not to do this.
\end{remark}
\begin{remark}
It is striking that the map between initial data and the function $W$, which we can interpret as scattering data, is so explicit and described by a simple integral transform.
For other integrable equations whose solutions are constructed via scattering theory and characterized by RH problems, like the Korteweg-de Vries equation, this map is far more involved, and simplifies to an integral transform only in small dispersion limits, see e.g.\ \cite{GP, LaxLevermore}.
\end{remark}
\begin{remark}
We believe that a more general class of solutions to \eqref{PDEsigma} can be constructed by considering Jànossy densities of the sine process instead of average multiplicative statistics. This is suggested by similar results for deformed Airy kernel determinants and the KdV equation \cite{CGRT}, and by the more general considerations in \cite{ClaeysGlesner}.
\end{remark}
\begin{remark}
Many integrable PDEs admit solutions involving Fredholm determinant expressions of the form $\det(1-K_{x,t})$, where $K_{x,t}$ depends on the variables $x,t$ and on scattering data. It is natural to ask whether our method to find an explicit relation between scattering data and solution data could apply to other integrable PDEs. The general principle to apply the method would consist of (1) finding a region in the $(x,t)$-plane where $K_{x,t}$ has small trace, such that 
$\log\det(1-K_{x,t})\sim -{\rm Tr}\,K_{x,t}$, (2) computing the trace of $K_{x,t}$ explicitly, and (3) inverting the expression for the trace to obtain an explicit expression for the scattering data in terms of the trace.
\end{remark}

\paragraph{Methodology and outline.}
In Section \ref{section:Fredholm}, we gather important properties of the Fredholm determinants $F_w(s)$ and $Q_W(y,s)$, which allow in particular to obtain their $s\to 0$ asymptotics, as well as that of their derivatives with respect to $s$ and $y$. Next, in Section \ref{section:RH}, we characterize the determinants $F_w(s)$ and $Q_W(y,s)$ in terms of a $2\times 2$ matrix-valued RH problem.
We also study $s\to 0$ asymptotics of the solution to the RH problem.
In Section \ref{section:s}, we use the RH characterization to study the $s$-dependence of $F_w(s)$ in detail, which allows us to relate $F_w(s)$ to a Zakharov-Shabat system, and to prove Theorem \ref{thm:s}. In Section \ref{section:ys}, we similarly study the $y$-dependence of $Q_W(y,s)$. In particular, we derive a Lax pair associated to the RH problem and use it to prove Theorem \ref{thm:q}. We emphasize that the Lax pair is of an unusual type, since it does not only contain derivatives in the deformation parameters $y$ and $s$, but also in the spectral variable $\lambda$ (i.e., the complex variable of the RH problem). The results in this section will yield the proof of Theorem \ref{thm:q}, and by combining it with the small $s$ asymptotics from Section \ref{section:Fredholm}, we also prove Theorem \ref{thm:assigma}.
In Section \ref{section:as}, we finally complete the proof of Theorem \ref{thm:scattering}.

\section{Preliminaries on the Fredholm determinants $F_w(s)$ and $Q_W(y,s)$}\label{section:Fredholm}
In this section, we will prove that the Fredholm determinants $F_w(s)$ and $Q_W(y,s)$ as well as their derivatives with respect to $s$ and $y$ are well-defined under our assumptions on $w$ and $W$, and we will establish their $s\to 0$ asymptotics.

\subsection{Trace-class operators}
We start by collecting useful properties of the integral operator $\chi_{I}\mathcal{K}^{\sin}_w\chi_{I}$ acting on $L^2(\mathbb R)$ with kernel 
$\chi_{I}(x)\mathrm{K}^{\sin}_w(x,y)\chi_{I}(y)$, 
\be 
\left(\chi_I\mathcal{K}^{\sin}_w\chi_I\right) f(x) = \chi_I(x)\int_I \mathrm{K}^{\sin}_w(x,y)f(y)dy,
\ee 
where $I$ is a compact real interval, $\chi_I$ is the indicator function of $I$, and
$\mathrm{K}^{\sin}_w$ is defined in \eqref{eq:deformedsinekernel}.

\begin{proposition}\label{prop:traceclass}
Let $I$ be a compact real interval.
\begin{itemize}
\item[(i)] For any integrable $w:\mathbb R\to\mathbb C$, the integral operator $\chi_{I}\mathcal{K}^{\sin}_w\chi_{I}:L^2\left(\mathbb R\right)\to L^2\left(\mathbb R\right)$ is a trace-class operator.
\item[(ii)] For any integrable $w:\mathbb R\to [0,1]$, the integral operator $\chi_{I}\mathcal{K}^{\sin}_w\chi_{I}:L^2\left(\mathbb R\right)\to L^2\left(\mathbb R\right)$ is Hermitian, and $0\leq \chi_I\mathcal{K}^{\sin}_w\chi_I< 1$; in particular $1- \chi_I\mathcal{K}^{\sin}_w\chi_I$ is invertible on $L^2(\mathbb R)$.
\end{itemize}
\end{proposition}
\begin{proof}
Let $w:\mathbb R\to\mathbb C$ be integrable. For the first part, we use the factorization \eqref{eq:sinekernel decomp deformed} to conclude that 
\[\chi_{I}\mathcal{K}^{\sin}_w\chi_{I}=(\chi_{I}\mathcal F^*\mathcal M_{\sqrt{w}})(\mathcal M_{\sqrt{w}}\mathcal F\chi_{I}),\]
for any choice of square root of $w$.
But it is not hard to see that 
$\chi_{I}\mathcal F^*\mathcal M_{\sqrt{w}}$ and $\mathcal M_{\sqrt{w}}\mathcal F\chi_{I}$ are Hilbert-Schmidt operators, hence $\chi_{I}\mathcal{K}^{\sin}_w\chi_{I}$ is trace-class on $L^2(\mathbb R)$.

\medskip

Next, suppose that $w$ takes values in $[0,1]$.
Since $0\leq \mathcal M_w\leq 1$, it follows from the factorization \eqref{eq:sinekernel decomp deformed} and the fact that the Fourier transform $\mathcal F$ is unitary that $0\leq \chi_I\mathcal{K}^{\sin}_w\chi_I\leq 1$.
Finally, to prove that $\chi_I\mathcal{K}^{\sin}_w\chi_I< 1$ and that $1-\chi_I\mathcal{K}^{\sin}_w\chi_I$ is invertible, it is enough to prove that $1$ is not an eigenvalue of $\chi_I\mathcal{K}^{\sin}_w\chi_I$. By contradiction, suppose that there exists a non-zero function $f\in L^2\left(\mathbb R\right)$ such that $\chi_I\mathcal{K}^{\sin}_w\chi_If=f.$ Then, $f$ would be supported inside $I$, and we would have 
\be 
\langle f,\chi_I\mathcal{K}^{\sin}_w\chi_I f\rangle  = \langle f,f\rangle. 
\ee 
But then from the one side, we would have by direct computation 
\be \langle f,\chi_I\mathcal{K}^{\sin}_w\chi_I f\rangle = \int_{\mathbb{R}}\Big|\mathcal{F}f(u)\Big|^2w(u)\d u\ee
while on the other side by  Plancherel's theorem
\be 
\langle f,f\rangle = \int_{\mathbb{R}} \Big|\mathcal{F}f(u)\Big|^2\d u.
\ee  
Consequently,
\be 
\int_{\mathbb{R}}\Big|\mathcal{F}f(u)\Big|^2 (1-w(u))\d u =0,
\ee 
and $\mathcal{F}f(u)=0$ for a.e.\ $u\in\mathbb R$ for which $w(u)\neq 1$.
Since $f$ is compactly supported, it follows from Paley-Wiener's theorem that $\mathcal Ff$ extends to an entire function, and this implies by the identity theorem that $\mathcal Ff=0$, and thus $f=0$, which is a contradiction. 
\end{proof}
As a consequence, we have that the Fredholm determinant $F_w(s)$ is well-defined for every $s>0$, provided that $w:\mathbb R\to\mathbb C$ is integrable. By replacing $w(\lambda)$ with $W(\lambda^2-y)$, we see that $Q_W(y,s)$ is well-defined for any $y\in\mathbb R$, $s>0$, provided that $\lambda\mapsto W(\lambda^2-y)$ is integrable for any $y\in\mathbb R$.
Moreover, if $w$ takes values in $[0,1]$, then $F_w(s)$ is non-zero for every $s>0$; if $W$ takes values in $[0,1]$, then $Q_W(y,s)$ is non-zero for every $s>0$, $y\in\mathbb R$.

\subsection{Expressions in terms of integrable operators}
Next, we prove that $F_w(s)$ and $Q_W(y,s)$ can be expressed as Fredholm determinants involving only the non-deformed sine kernel operator and a multiplication operator, more precisely we prove \eqref{eq:def Fw-alt} and \eqref{eq:def QW-alt}. Similar identities are well-known for other operators, see e.g.\ \cite{AmirCorwinQuastel, Krajenbrink}, and the proof below contains no conceptual novelties, but is included for the convenience of the reader.

\begin{proposition}
	\label{prop: fred det equalities}
\begin{itemize}
\item[(i)] Let $w:\mathbb R\to\mathbb C$ be integrable.
Then,
\be \label{eq:Fdet2}F_w(s)=\det\left(1-\sqrt{w_{0,s}}\mathcal K^{\sin}\sqrt{w_{0,s}}\right),\ee
where $w_{0,s}(u)=w\left(\frac{\pi u}{s}\right)$.
\item[(ii)] Let $y\in\mathbb R$ and $W:\mathbb R\to\mathbb C$ be such that $W(.^2-y)$ is integrable.
Then,
\be\label{eq:Qdet2} Q_W(y,s)=\det\left(1-\sqrt{w_{y,s}}\mathcal K^{\sin}\sqrt{w_{y,s}}\right),\ee
where $w_{y,s}$ is given by \eqref{def:wst}.
\end{itemize}
\end{proposition}
\begin{proof}
It sufficient to prove \eqref{eq:Fdet2}, since \eqref{eq:Qdet2} then follows upon setting $w(u)=W(u^2-y)$.
By \eqref{eq:FW2} and \eqref{eq:sinekernel decomp deformed}, we have
\be F_w(s)=\det\left(1-\chi_{(-1/2,1/2)}
\mathcal F^* \sqrt{w_{0,s}}.\sqrt{w_{0,s}}\mathcal F\chi_{(-1/2,1/2)}
\right).\label{eq:Fwdecomposed}\ee
Now we verify that 
 $\chi_{(-1/2,1/2)}
\mathcal F^* \sqrt{w_{0,s}}$ and $\sqrt{w_{0,s}}\mathcal F\chi_{(-1/2,1/2)}$ are Hilbert-Schmidt operators, and use the property $\det(1-AB)=\det(1-BA)$ for Hilbert-Schmidt operators $A,B$, such that
\[F_w(s)=\det\left(1-\sqrt{w_{0,s}}\mathcal F\chi_{(-1/2,1/2)}
\mathcal F^* \sqrt{w_{0,s}}
\right)=\det\left(1-\sqrt{w_{0,s}}\mathcal K^{\sin} \sqrt{w_{0,s}}
\right).\]
For the last equality we used \eqref{eq:sinekernel decomp0}.
\end{proof}

The advantage of the representations \eqref{eq:Fdet2}--\eqref{eq:Qdet2} is that they involve operators with integrable kernels. Indeed, we can write
\be 
\label{eq: sym kernel}
{\sqrt{w_{y,s}(u)}} \mathrm{K}^{\sin}(u,v){\sqrt{w_{y,s}(v)}} = \frac{\vec{f}(u)^\top \vec{g}(v)}{u-v},
\ee
where
\be
\label{eq:vectors fg}
\vec{f}(u) = \frac{\sqrt{w_{y,s}(u)}}{2\pi \i}\begin{pmatrix}
	\e^{\i\pi u}\\\e^{-\i\pi u} 
\end{pmatrix},\;\;\vec{g}(v) = \sqrt{w_{y,s}(v)}\begin{pmatrix}
	\e^{-\i\pi v}\\-\e^{\i\pi v} 
\end{pmatrix}.
\ee
This integrable representation of the kernel will enable us later to characterize $F_w(s)$ in terms of a $2\times 2$ matrix RH problem, and to derive underlying differential equations.

\subsection{Small $s$ asymptotics}

The following result describes the small $s$ behavior of the  Fredholm determinants $Q_W(y,s)$ and their logarithmic derivatives with respect to $s$ and $y$. Below, we write as before $\sigma_W(y,s)=\log Q_W(y,s)$.
\begin{proposition}\label{prop:sto0}
Let $y\in\mathbb R$ and let $W:\mathbb R\to\mathbb C$ be such that $W(.^2-y)$ is a Schwartz function for every $y\in\mathbb R$.
Then,
\be\label{asTrace}\lim_{s\to 0}\frac{1}{s}\sigma_W(y,s)=-\frac{2}{\pi}\int_0^{\infty}W(u^2-y)\d u,\ee
and
\be\label{eq:s0}\lim_{s\to 0}s\partial_s\sigma_W(y,s)=0,\qquad \lim_{s\to 0}\partial_y\sigma_W(y,s)=0.\ee
\end{proposition}
\begin{proof}
We first recall that
\be\label{Lys}Q_W(y,s)=\det(1-\mathcal L_{y,s}),\qquad \mathcal L_{y,s}:=\chi_{(-1/2,1/2)}\mathcal F^*\sqrt{w_{y,s}}.\sqrt{w_{y,s}}
\mathcal F \chi_{(-1/2,1/2)},\ee and then prove that the trace-norm of $\mathcal L_{y,s}$ is $O(s)$ as $s\to 0$.
Both factors in the above factorization are Hilbert-Schmidt operators. Hence, the trace-norm is bounded by the product of Hilbert-Schmidt norms: 
\[\|\mathcal L_{y,s}\|_1\leq \left\|\chi_{(-1/2,1/2)}
\mathcal F^* \sqrt{w_{y,s}}\right\|_2\ \left\|\sqrt{w_{y,s}}\mathcal F \chi_{(-1/2,1/2)}\right\|_2,\]

and it is straightforward to estimate the Hilbert-Schmidt norms on the right: for the first one, we have
\begin{multline*}\left\|\chi_{(-1/2,1/2)}
\mathcal F^* \sqrt{w_{y,s}}\right\|_2^2\leq \int_{-\infty}^{+\infty}\int_{-1/2}^{1/2}|w_{y,s}(y)|\d x\d y=\|w_{y,s}\|_1\\
=2\int_{0}^\infty W\left(\frac{\pi^2r^2}{s^2}-y\right)\d r=\frac{2s}{\pi}\int_{0}^\infty W\left(\lambda^2-y\right)\d \lambda. \end{multline*}
Similar estimates hold for the second factor, such that  the operator norm of 
$\mathcal L_{y,s}$ is $O(s)$ as $s\to 0$. 
Hence, as $s\to 0$, we have
\[Q_W(y,s)=1-{\rm Tr}\,\mathcal L_{y,s}+O(s^2),\qquad s\to 0.\]
Taking the logarithm and computing the trace explicitly as
\[{\rm Tr}\,\mathcal L_{y,s}=\frac{2s}{\pi}\int_0^{\infty}W(\lambda^2-y)\d \lambda,\]
we obtain \eqref{asTrace}.

Moreover, the operator 
$1-\mathcal L_{y,s}$ is invertible for small $s$ and by a Neumann expansion, we have the operator norm bound 
\be\label{eq:estimateinverse}\|\left(1-\mathcal L_{y,s}\right)^{-1}\|\leq 2,\ee
for $s$ sufficiently small.

\medskip

Then by Jacobi's identity and similar norm estimates as before, we obtain for the $s$-derivative
\begin{align*}\left|\partial_s\log Q_W(y,s)\right|&=\left|{\rm Tr}\left[\left(\partial_s\mathcal L_{y,s}\right)\left(1-\mathcal L_{y,s}\right)^{-1}\right]\right|\\
&\leq 2
\ \left\|{\partial_sw_{y,s}}\right\|_1\\
&=8\int_0^{\infty}\left|W'\left(\frac{\pi^2r^2}{s^2}-y\right)\frac{\pi^2 r^2}{s^3}\right|\d r\\&=\frac{8}{\pi}\int_0^{\infty}|W'(\lambda^2-y)|\lambda^2\d \lambda
.\end{align*}

\medskip

For the $y$-derivative, we similarly obtain
\[\left|\partial_y\log Q_W(y,s)\right|\leq 2
\ \left\|{\partial_yw_{y,s}}\right\|_1=\frac{4s}{\pi}\int_0^{\infty}|W'(\lambda^2-y)|\d \lambda.\]
The results in \eqref{eq:s0} now follow directly  since $\sigma_W=\log Q_W$.
\end{proof}

\section{RH characterization of $F_w(s)$ and $Q_W(y,s)$}\label{section:RH}
Our objective here is to characterize $F_w(s)$ and $Q_W(y,s)$ in terms of the unique solution of a $2\times 2$ matrix RH problem.
To that end, we closely follow the method developed in \cite{CafassoClaeys, CafassoClaeysRuzza} for the study of deformed Airy kernel determinants, relying on the more general theory of \cite{BertolaCafasso, DIZ, IIKS}. 

\subsection{Characterization by a RH problem}
\label{subsec:characterization by a RH problem}
We will express $F_w(s)$ and $Q_W(y,s)$ in terms of the solution $U(\lambda)=U(\lambda;s)$ of the following RH problem, which depends on a parameter $s>0$ and on a function $w:\mathbb R\to \mathbb C$, which we will assume to be Schwartz.
\paragraph{RH problem for $U$.}  
\begin{itemize}
	\item[(1)] $U:\C\setminus\mathbb R \rightarrow GL(2,\C)$ is analytic.
	\item[(2)]   $U(\lambda;s)$ has continuous boundary values  $U_{\pm}(\lambda;s)$ when $\lambda$ approaches $\R$ from either above ($+$) or below ($-$) and they satisfy the jump condition 
	\be
	U_{+}(\lambda;s)=U_{-}(\lambda;s) \underbracket{\begin{pmatrix}
			1 & 1-w(\lambda)\\
			0&1
	\end{pmatrix}}_{=J_{U}(\lambda)}, \;\; \lambda \in \R.
	\ee
	\item[(3)] There exists a matrix $U_1=U_1(s)$ such that we have as $\lambda\to\infty$,
	\be
	\label{eq:asy U}
	U(\lambda;s) = \left(I +\frac{U_1}{\lambda}+O(\lambda^{-2})\right) \e^{\i s \lambda \sigma_3}\begin{cases}
		\begin{pmatrix}
			1&1\\1&0
		\end{pmatrix}, & \Im \lambda >0,\\ 
		\\
		\begin{pmatrix}
			1&0\\
			1&-1
		\end{pmatrix},& \Im \lambda <0,
	\end{cases}
	\ee
	where $I=\begin{pmatrix}1&0\\0&1\end{pmatrix}$ and $\sigma_3=\begin{pmatrix}1&0\\0&-1\end{pmatrix}$.
\end{itemize}
The RH problem and its solution depend on the choice of $w$. When this dependence will become important, we will write $U^{(w)}$ instead of $U$ for the solution of the RH problem.

The next subsection will be devoted to the proof of the following result, which essentially follows from the theory of integrable kernel operators developed by Its, Izergin, Korepin, and Slavnov \cite{IIKS}.

\begin{proposition}\label{prop:detU}
\begin{itemize}
\item[(i)] Let $w:\mathbb R\to\mathbb C$ be a Schwartz function. The RH problem for $U^{(w)}$ is uniquely solvable if and only if $F_w(s)\neq 0$, and in that case we have the identity
\be
	\label{eq:Fw log der}
\partial_s\log F_w(s)= \frac{1}{2\pi\i s} \int_{\R} \left[U_+^{(w)}(\lambda;s)^{-1}\frac{\d}{\d \lambda}U_+^{(w)}(\lambda;s)\right]_{2,1}w'(\lambda)\lambda\d \lambda.\ee
\item[(ii)] Let $y\in\mathbb R$ and $W:\mathbb R\to\mathbb C$ be  such that $W(.^2-y)$ is a Schwartz function. The RH problem for $U^{(w)}$ is uniquely solvable if and only if $Q_W(y,s)\neq 0$, and in that case we have the identity
\be
 \label{eq:QW log der}
\partial_s\log Q_W(y,s)=\frac{1}{2\pi\i s} \int_{\R} \left[U_+^{(w)}(\lambda;s)^{-1}\frac{\d}{\d \lambda}U_+^{(w)}(\lambda;s)\right]_{2,1} w'(\lambda)\lambda\d \lambda.
\ee
\end{itemize}\end{proposition}

Observe that part (ii) of the above result is a simple consequence of part (i), obtained by setting $w(\lambda)=W(\lambda^2-y)$. The next subsection will be devoted to the proof of part (i).

\subsection{Proof of Proposition \ref{prop:detU}}

We will first relate $F_w(s)$ to another RH problem. Indeed, the RH problem for $U$ is obtained by transformations of the standard RH problem obtained using the Its-Izergin-Korepin-Slavnov method from \cite{IIKS}, associated to the integrable kernel written in equations \eqref{eq: sym kernel}--\eqref{eq:vectors fg}, which is described below.
As before, we let $s>0$ and $w:\mathbb R\to\mathbb C$ be a Schwartz function.
Let $\vec f, \vec g$ be as in \eqref{eq:vectors fg} with $y=0$.

Consider the following RH problem for $Y(\zeta)=Y(\zeta;s)$.
\paragraph{RH problem for $Y$.}  
\begin{itemize}
	\item[(1)]  $Y:\C\setminus\mathbb R \rightarrow GL(2,\C)$ is analytic.
	\item[(2)]   $Y(\zeta;s)$ has continuous boundary values  $Y_{\pm}(\zeta;s)$ when approaching $\R$ from above or below, they are related by
	\be
	Y_{+}(\zeta;s)=Y_{-}(\zeta;s) \underbracket{\left( I-2\pi\i \vec{f}(\zeta)\vec{g}(\zeta)^{\top}\right)}_{=J_Y(\zeta;s)}, \;\; \zeta \in \R,
	\ee
	for $\vec{f}, \vec{g}$ given in \eqref{eq:vectors fg} with $y=0$.
	\item[(3)] As $\zeta\rightarrow \infty$, we have
	\be
	Y(\zeta;s) = I +O(\zeta^{-1}).
	\ee
\end{itemize}
From the general theory developed in \cite{DIZ, IIKS}, it follows that the solution of this RH problem exists and is unique if and only if $F_w(s)\neq 0$, and it is then related to the operator $\mathcal{R}$ defined by
\begin{align}
\mathcal{R}&=\left(1-\sqrt{w_{s}}\mathcal{K}^{\sin}\sqrt{w_{s}}\right)^{-1}-1\\
&= \sqrt{w_{s}}\mathcal{K}^{\sin}\sqrt{w_{s}}\left(1-\sqrt{w_{s}}\mathcal{K}^{\sin}\sqrt{w_{s}}\right)^{-1}\\
&=\left(1-\sqrt{w_{s}}\mathcal{K}^{\sin}\sqrt{w_{s}}\right)^{-1}\sqrt{w_{s}}\mathcal{K}^{\sin}\sqrt{w_{s}},
\end{align}
where we denote for brevity $w_s(u)=w_{0,s}(u)=w\left(\frac{\pi u}{s}\right)$.
In particular, $\mathcal{R}$ has a kernel 
\be
R(u,v)=\frac{\vec{F}(u)^\top \vec{G}(v)}{u-v},
\ee
where $\vec F$, $\vec G$ are characterized by the RH solution $Y$,
\be
\vec{F}(u)=Y_+ (u)\vec{f}(u), \qquad \vec{G}(v)= Y_+^{-\top}(v) \vec{g}(v),
\ee
or equivalently
\be
\vec{F}(u)=\left((1-\sqrt{w_{s}}\mathcal{K}^{\sin}\sqrt{w_{s}})^{-1}\vec{f} \;\right)(u), \; \vec{G}(v)= \left((1-\sqrt{w_{s}}\mathcal{K}^{\sin}\sqrt{w_{s}})^{-1}\vec{g}\; \right) (v).
\ee
As a byproduct, we have 
\be\label{eq:resolventkernel}
R(u,v)=\frac{\vec{f}(u)^\top Y_+^\top(u) Y_+^{-\top}(v)\vec{g}(v)}{u-v}.
\ee
To compute $R(\zeta,\zeta)$ on the diagonal, we need to be able to differentiate $Y_+(\zeta)$ with respect to $\zeta$. Since $w$ is smooth and fast decaying at $\pm\infty$, the jump matrix is smooth and close to $I$ at $\pm\infty$. It then follows from the general theory of RH problems, see e.g.\ \cite{DIZ}, that the boundary values $Y_\pm$ are indeed differentiable, such that
\be
\label{eq:R diagonal kernel}
R(\zeta,\zeta)=\frac{\d}{\d \zeta}\left(\vec{f}\;^\top(\zeta)Y_+^\top(\zeta)\right)Y_+^{-\top}(\zeta)\vec{g}(\zeta).
\ee
and we can also verify from this expression that $\frac{R(\zeta,\zeta)}{w_s(\zeta)}$ is integrable on the real line.

\medskip

With a first transformation of the RH problem for $Y$, we are able to simplify the jump matrix. Notice that the jump matrix $J_Y(\zeta;s)$ is explicitly given by 
\be 
J_Y(\zeta;s)=
\begin{pmatrix}
	1-w_{s}(\zeta) & w_{s}(\zeta)\e^{2\pi \i \zeta }\\
	-w_{s}(\zeta)\e^{-2\pi\i\zeta} & 1+w_{s}(\zeta)
\end{pmatrix}.
\ee
In order to transform this jump matrix to an upper-triangular jump matrix, we define
\be\label{def:Psi}
\Psi(\zeta; s) \coloneqq \begin{cases}
	Y(\zeta;s) \begin{pmatrix}
		\e^{\pi\i\zeta} & \e^{\pi\i\zeta}\\
		\e^{-\pi \i \zeta}&0
	\end{pmatrix}, &\Im \zeta >0, \\ 
	\\
	Y(\zeta;s) \begin{pmatrix}
		\e^{\pi\i\zeta} & 0\\
		\e^{-\pi \i \zeta}&-\e^{-\pi\i\zeta}
	\end{pmatrix}, &\Im \zeta <0.
\end{cases}
\ee
We then obtain the following RH problem for $\Psi(\zeta)=\Psi(\zeta;s)$.
 \paragraph{RH problem for $\Psi$.}  
\begin{itemize}
	\item[(1)] $\Psi:\C\setminus\mathbb R \rightarrow GL(2,\C)$ is analytic.
	\item[(2)]   $\Psi(\zeta;s)$ has continuous boundary values  $\Psi_{\pm}(\zeta;s)$ when approaching $\R$ from either left or right, and
	\be
	\Psi_{+}(\zeta;s)=\Psi_{-}(\zeta;s) \underbracket{\begin{pmatrix}
			1 & 1-w_{s}(\zeta)\\
			0&1
	\end{pmatrix}}_{=J_{\Psi}(\zeta;s)}, \;\; \zeta \in \R.
	\ee
	\item[(3)] As $\zeta\rightarrow \infty$, we have
	\be
	\Psi(\zeta;s) = \left(I +O(\zeta^{-1})\right) \e^{\pi \i\zeta \sigma_3}\begin{cases}
		\begin{pmatrix}
			1&1\\1&0
		\end{pmatrix}, & \Im \zeta >0,\\ 
		\\
		\begin{pmatrix}
			1&0\\
			1&-1
		\end{pmatrix}, & \Im \zeta <0.
	\end{cases}
	\ee
\end{itemize}
Indeed, we have
\be 
\begin{pmatrix}
	\e^{-\i\pi\zeta}&0\\
	\e^{-\i\pi\zeta}&-\e^{\i\pi\zeta}
\end{pmatrix}J_{Y}(\zeta;s)\begin{pmatrix}
	\e^{\i\pi\zeta}&\e^{\i\pi\zeta}\\
	\e^{-\i\pi\zeta}&0
\end{pmatrix} = J_{\Psi}(\zeta;s),
\ee
implying that $\Psi_+(\zeta;s) = \Psi_-(\zeta;s)  J_{\Psi}(\zeta;s)$. Also notice that
\be 
\e^{\pi \i \zeta \s_3} \begin{pmatrix}
	1&1\\1&0
\end{pmatrix} = \begin{pmatrix}
	\e^{\pi\i\zeta} & \e^{\pi\i\zeta}\\
	\e^{-\pi \i \zeta}&0
\end{pmatrix}, \; \; \; \e^{\pi \i \zeta \s_3} \begin{pmatrix}
	1&0\\1&-1
\end{pmatrix} = \begin{pmatrix}
	\e^{\pi\i\zeta} & 0\\
	\e^{-\pi \i \zeta}&-\e^{-\pi\i\zeta}
\end{pmatrix},
\ee
thus the asymptotic condition for $\Psi(\zeta;s)$ holds.

\medskip

Finally, we define
\be\label{def:U}
U(\lambda;s)
\coloneqq \Psi(\zeta = \frac{s}{\pi}\lambda;s),
\ee
and now it is straightforward to check that $U(\lambda;s)$ solves the RH problem for $U$ given at the beginning of this section. 

\medskip

We can finally prove Proposition \ref{prop:detU}.
\begin{proof}
	We use the representation for $F_w(s)$ given in \eqref{eq:Fdet2}, together with Jacobi's identity and the formula for the kernel $R(\zeta,\zeta)$ given in \eqref{eq:R diagonal kernel} to compute the $s$-logarithmic derivative of $F_w(s)$:
	\begin{align}
		\pa_s \log F_w(s)&= \pa_s \log \det \left(1-\sqrt{w_{s}}\mathcal{K}^{\sin}\sqrt{w_{s}}\right)\nonumber\\
		&= -{\rm Tr} \left[ \left(1-\sqrt{w_{s}}\mathcal{K}^{\sin}\sqrt{w_{s}}\right)^{-1}\pa_s (\sqrt{w_{s}}\mathcal{K}^{\sin}\sqrt{w_{s}})\right]\nonumber\\
		&= -\int_{\R}  \frac{\pa_s w_{s}(\zeta)}{w_{s}(\zeta)} R(\zeta,\zeta) \d \zeta\nonumber\\
		&=-\int_{\R}  \frac{\pa_s w_{s}(\zeta)}{w_{s}(\zeta)} \frac{\d}{\d \zeta}\vec{f}\;^\top(\zeta)\vec{g}(\zeta)\d \zeta -\int_{\R}  \frac{\pa_s w_{s}(\zeta)}{w_{s}(\zeta)} \vec{f}\;^\top(\zeta) \frac{\d}{\d \zeta}Y_+^\top(\zeta)Y_+^{-\top}(\zeta)\vec{g}(\zeta)\d \zeta\nonumber.
	\end{align}
Note that the integral in the third line converges because  $R(\zeta,\zeta)/w_s(\zeta)$ is integrable.
	The first term gives 
	\begin{align}
		&-\int_{\R}  \frac{\pa_s w_{s}(\zeta)}{w_{s}(\zeta)} \frac{\d}{\d \zeta}\vec{f}\;^\top(\zeta)\vec{g}(\zeta)\d \zeta = -\int_{\R}  \pa_s w_{s}(\zeta)\d \zeta.\label{firstterm}
	\end{align}
	To compute the second term, we use the transformation 
	\be
	U_+(\lambda)= Y_+(\lambda s/\pi) 
	\underbracket{\begin{pmatrix}
			\e^{\i\lambda s} & \e^{\i\lambda s}\\
			\e^{-\i \lambda s}&0
	\end{pmatrix}}_{=\Phi(\lambda s/\pi)}
	\ee
	to replace $Y$ inside the integral. We first take the transpose and then change the variable $\zeta=s\lambda /\pi$, to obtain 
	\begin{align*}
		&-\int_{\R}  \frac{\pa_s w_{s}(\zeta)}{w_{s}(\zeta)} \vec{f}\;^\top(\zeta) \frac{\d}{\d v}Y_+^\top(\zeta)Y_+^{-\top}(\zeta)\vec{g}(\zeta)\d \zeta \\&\qquad=-\int_{\R}  \frac{(\pa_s w_{s})(s\lambda/\pi)}{w_{s}(s\lambda/\pi)} \vec{g}^\top(s\lambda/\pi)Y_+^{-1}(s\lambda/\pi) \frac{\d}{\d \lambda}Y_+^\top(s\lambda/\pi)\vec{f}(s\lambda/\pi)\d \lambda\\
		&\qquad= -\int_{\R}  \frac{(\pa_s w_{s})(s\lambda/\pi)}{w_{s}(s\lambda/\pi)} \vec{g}^\top(s\lambda/\pi)\Phi(s\lambda/\pi) \frac{\d}{\d \lambda}\Phi^{-1}(s\lambda/\pi)  \vec{f}(s\lambda/\pi)\d \lambda\\&\qquad\qquad -\int_{\R}  \frac{(\pa_s w_{s})(s\lambda/\pi)}{w_{s}(s\lambda/\pi)} \vec{g}^\top(s\lambda/\pi) \Phi(s\lambda/\pi)U_+^{-1}(\lambda) \frac{\d}{\d \lambda}U_+(\lambda) \Phi(s\lambda/\pi)^{-1} \vec{f}(s\lambda/\pi)\d \lambda.
	\end{align*}
	The first term here gives 
	\be
 -\int_{\R}  \frac{(\pa_s w_{s})(s\lambda/\pi)}{w_{s}(s\lambda/\pi)} \vec{g}^\top(s\lambda/\pi)\Phi(s\lambda/\pi) \frac{\d}{\d \lambda}\Phi^{-1}(s\lambda/\pi)  \vec{f}(s\lambda/\pi)\d \lambda
 = \int_{\R}  (\pa_s w_{s})(\zeta)\d \zeta,
	\ee
	thus it cancels out with the term obtained in \eqref{firstterm}. What is left, is given by (see \eqref{eq:vectors fg})
	\begin{align*}
		&-\int_{\R}  \frac{(\pa_s w_{s})(s\lambda/\pi)}{w_{s}(s\lambda/\pi)} \vec{g}^\top(s\lambda/\pi) \Phi(s\lambda/\pi)U_+^{-1}(\lambda) \frac{\d}{\d \lambda}U_+(\lambda) \Phi(s\lambda/\pi)^{-1} \vec{f}(s\lambda/\pi)\d \lambda  \\
		&= -\int_{\R}  \frac{(\pa_s w_{s})(s\lambda/\pi)}{2\pi \i}(0,1)U_+^{-1}(\lambda) \frac{\d}{\d \lambda}U_+(\lambda) \begin{pmatrix}
			1\\0
		\end{pmatrix}\d \lambda \\
		&= -\int_{\R}  \frac{(\pa_s w_{s})(s\lambda/\pi)}{2\pi \i} \left[ U_+^{-1}(\lambda)\frac{\d}{\d \lambda}U_+(\lambda)\right]_{2,1} \d \lambda.
	\end{align*}
	Now we compute
	\be
	\pa_s w_{s}(\zeta)|_{\zeta=s\lambda/\pi}=  -\frac{\lambda}{s} w'(\lambda),
	\ee
and substitute this into the last integral, to arrive at \eqref{eq:Fw log der}.
\end{proof}
{\subsection{Analytical properties of the RH problem}}
Later on, we will need to differentiate the RH solution $U(\lambda)=U^{(w)}(\lambda;s)$ with respect to the parameter $s>0$, and also with respect to $y\in\mathbb R$ if $w(\lambda)=W(\lambda^2-y)$. In addition, we will need to differentiate the asymptotic expansion \eqref{eq:asy U} with respect to $y, s$, and $\lambda$.
To justify that we can differentiate the solution and the asymptotic expansion, we rely on the well understood functional analysis behind RH problems, developed in \cite{DIZ}. For this, it is important that the jump matrices are sufficiently smooth and decay sufficiently fast at $\pm\infty$, and this is the origin of our requirement that $w$ is a Schwartz function.
Without entering into too much technical detail, let us outline how this general theory of RH problems applies to our RH problem.
Recall the RH problems for $U$ and $Y$, and define 
\be\label{def:T}T(\lambda)=T(\lambda;y,s)=Y^{(w)}(s\lambda/\pi;s)=U^{(w)}(\lambda;s)\begin{cases}
		\begin{pmatrix}
			0&1\\1&-1
		\end{pmatrix}, & \Im \lambda >0,\\ 
		\\
		\begin{pmatrix}
			1&0\\
			1&-1
		\end{pmatrix},& \Im \lambda <0,
	\end{cases}\e^{-\i s\lambda\sigma_3},\ee with $w(u)=W(u^2-y)$. Then $T$ is analytic in $\mathbb C\setminus \mathbb R$, $T(\lambda)\to I$ as $\lambda\to\infty$, and
\[T_+(\lambda)=T_-(\lambda)\underbracket{\begin{pmatrix}1-w(\lambda)&w(\lambda)\e^{2\i s\lambda}\\
-w(\lambda)\e^{-2\i s\lambda}&1+w(\lambda)\end{pmatrix}}_{=J_T(\lambda)}.\]
These RH properties can be reformulated as an integral equation, namely we have
\[T(\lambda)=I+\frac{1}{2\pi\i}\int_{-\infty}^{+\infty}T_-(\xi)(J_T(\xi)-I)\frac{\d\xi}{\xi-\lambda},\]
and the boundary value $T_-$ satisfies the integral equation
\[(T_{-}-I)=C_-\left[(T_--I)(J_T-I)\right]+C_-[J_T-I],\]
where $C_-$ is the Cauchy transform defined as
\[C_-[f](\lambda)=\frac{1}{2\pi\i}\lim_{\epsilon\to 0_+}\int_{-\infty}^{+\infty}f(\xi)\frac{\d\xi}{\xi-\lambda+\i\epsilon}.\]
Alternatively, writing $C_{J_T}[f]=C_-[f(J_T-I)]$, we have
\be\label{eq:inteqT}T_--I=(1-C_{J_T})^{-1}\left[C_-[J_T-I]\right],\ee
and the invertibility of the operator $1-C_{J_T}$ is equivalent to the solvability of the RH problem (see e.g.\ \cite[Section 2]{DIZ}.
The above equation can be differentiated with respect to $s$
to obtain
\be\label{eq:intT}\partial_s T_-=(1-C_{J_T})^{-1}\left[C_-[(T_--I)\partial_s J_T]+C_-\partial_s J_T\right],\ee
provided that $J_T$ is smooth and $J_T$ as well as its $s$-derivative are sufficiently fast decaying as $\lambda \to\pm\infty$.
Taking the derivative with respect to $y$ works similarly. Expanding the last equation as $\lambda \to\infty$, one shows that the expansion 
\[T(\lambda)=I+\frac{U_1}{\lambda}+O(\lambda^{-2}),\quad U_1=-\frac{1}{2\pi\i}\int_{\mathbb R}T_-(\xi)(J_T-I)(\xi)\d\xi, 
\]
 continues to hold after differentiating it with respect to $s$ or $y$. In particular we have
 \[\partial_s T(\lambda)=\frac{\partial_s U_1}{\lambda}+O(\lambda^{-2}),\qquad \partial_s U_1=-\frac{1}{2\pi\i}\int_{\mathbb R}\partial_s\left(T_-(\xi)(J_T-I)(\xi)\right)\d\xi,
\]
and similarly for the $y$-derivative.

For the $\lambda$-derivative, the argument is different: here we can see that the RH conditions for $T$ imply RH conditions for $T'$, which yield
\[T'(\lambda)=\frac{1}{2\pi\i}\int_{-\infty}^{+\infty}\left(T_-'(\xi)(J_T(\xi)-I)+T_-(\xi)J_T'(\xi)\right)\frac{\d\xi}{\xi-\lambda},\]
and
\be\label{eq:intT}T_-'=(1-C_{J_T})^{-1}\left[C_-[T_-J_T']\right],\ee
and expanding this as $\lambda\to\infty$, we find that
\[T'(\lambda)=-\frac{U_1}{\lambda^2}+O(\lambda^{-3}). 
\]
We can similarly take higher order partial derivatives of $T$ or $U$, and accordingly differentiate expansion \eqref{eq:asy U} multiple times, provided that $J_T$ is smooth and decaying sufficiently fast at $\pm\infty$, which is the case since $w$ is a Schwartz function. In what follows, we will sometimes tacitly rely on arguments like the above to justify taking derivatives and derivatives of asymptotic expansions.

{We conclude this Section with some results on the small $s$ behavior of the large $\lambda$ asymptotic coefficient $U_1$. They will be used in the proof of Theorem \ref{thm:q} in Section \ref{section:ys}.}
\begin{proposition}\label{prop:RHsto0}
Let $W:\mathbb R\to\mathbb C$ be such that $W(.^2-y)$ is a Schwartz function for every $y\in\mathbb R$. Then, for any $y\in\mathbb R$, we have
\be\lim_{s\to 0}sU_1(y,s)=0,\qquad \lim_{s\to 0}\partial_y U_1(y,s)=0.\ee
\end{proposition}
\begin{proof}
Let $T$ be as in \eqref{def:T}. For $s=0$, we can construct an explicit solution to the RH conditions for $T$ and we call it $T_0$. We construct
$T_0$ analytic off the real line such that $T_0(\lambda)\to I$ as $\lambda\to\infty$, and
\[T_{0,+}(\lambda)=T_{0,-}(\lambda)\underbracket{\begin{pmatrix}1-w(\lambda)&w(\lambda)\\
-w(\lambda)&1+w(\lambda)\end{pmatrix}}_{=J_{T_0}(\lambda)},\]
as follows:
\[T_0(\lambda)=I+\frac{1}{2\pi \i}\int_{-\infty}^{+\infty}w(\xi)\frac{\d\xi}{\xi-\lambda}\begin{pmatrix}-1&1\\-1&1\end{pmatrix}.\]
Next we set
\[R(\lambda)=T(\lambda)T_0(\lambda)^{-1}.\]
Then $R$ solves a small norm RH problem as $s\to 0$, in the sense that 
$R$ is analytic in $\mathbb C\setminus\mathbb R$, that $R(\lambda)\to I$ as $\lambda\to\infty$, and that 
\[R_+(\lambda)=R_-(\lambda)J_R(\lambda),\qquad\lambda\in\mathbb R,\]
where $J_R(\lambda)= T_{0,-}(\lambda)J_T(\lambda)J_{T_0}(\lambda)^{-1}T_{0,-}^{-1}(\lambda)$, such that $J_R-I$ is $O(s)$ as $s\to 0$, in $L^2$-, $L^1$-, and $L^\infty$- norms. 
Indeed, for any of these norms combined with a sub-multiplicative norm on $2\times 2$ matrices, since $T_{0,-}$ and $T_{0,-}^{-1}$ are uniformly bounded, there exists a constant $C>0$ such that
\be 
\| J_R-I \| \leq C \| J_{T}J_{T_0}^{-1} - I\|,
\ee 
and
\be
J_{T}(\lambda)J_{T_0}^{-1}(\lambda) - I =\begin{pmatrix}
w(\lambda )^2(\e^{2 \i \lambda s}-1)&w(\lambda)(1-w(\lambda))(\e^{2\i\lambda s}-1)\\
w(\lambda)(1+w(\lambda))(1-\e^{-2\i\lambda s})&w(\lambda )^2(\e^{-2 \i \lambda s}-1)
\end{pmatrix}, 
\ee 
with $w(\lambda)=W(\lambda^2-y)$,
whose entries are, for any $y\in\mathbb R$, 
$O(s\lambda w(\lambda))$
as $s\to 0$, uniformly for $\lambda\in\mathbb R$.
Similarly as outlined before, it then follows that $R_{-}-I$ satisfies the integral equation
\[(R_{-}-I)=C_-\left[(R_--I)(J_R-I)\right]+C_-[J_R-I],\]
hence
\[(R_{-}-I)=\left(1-C_{J_R}\right)^{-1}\left[C_-[J_R-I]\right],\quad\text{where}\quad C_{J_R}[h]=C_-[h(J_R-I)].\]
Since $C_{J_R}$ has small operator norm, the right hand side can now be written as a Neumann series.
This integral equation can also be differentiated with respect to $y$, in order to obtain
\[\partial_y R_-=(1-C_{J_R})^{-1}\left[C_-[(R_--I)\partial_y J_R]+C_-\partial_y J_R\right],\]
provided that $J_R$ is smooth and decays sufficiently fast at $\pm\infty$.

As $\lambda\to\infty$, we obtain
\[R(\lambda)=I+\frac{R_1}{\lambda}+O(\lambda^{-2}),
\]
and 
\[\partial_y R(\lambda)=\frac{\partial_y R_1}{\lambda}+O(\lambda^{-2}),\qquad \partial_y R_1=-\frac{1}{2\pi\i}\int_{\mathbb R}\partial_y\left(R_-(\xi)(J_R-I)(\xi)\right)\d\xi.
\]

Using \eqref{def:T} and the fact that $T=RT_0$, we obtain
\be\label{eq:UR}U^{(w)}(\lambda;s)=R(\lambda;y,s)T_0(\lambda)\e^{\i s \lambda \sigma_3}\begin{cases}
		\begin{pmatrix}
			1&1\\1&0
		\end{pmatrix}, & \Im \lambda >0,\\ 
		\\
		\begin{pmatrix}
			1&0\\
			1&-1
		\end{pmatrix},& \Im \lambda <0.
	\end{cases}\ee
As $s\to 0$, it is straightforward to check that $R_1$ and $\partial_y R_1$ are $O(s)$, since $J_R-I$ is $O(s)$ in $L^\infty$-, $L^1$-, and $L^2$-norms.
Expanding \eqref{eq:UR} as $\lambda\to\infty$, we obtain that $U_1(y,s)=O(1)$ as $s\to\infty$. Differentiating with respect to $y$, we obtain
that $\partial_y U_1(y,s)=\partial_y R_1(y,s)=O(s)$ as $s\to 0$. 
\end{proof}

\section{Differential equations in $s$}\label{section:s}
Next, we derive a differential equation for the matrix-valued function $U(\lambda;s)$ with respect to the parameter $s$. This will allow us to prove Theorem \ref{thm:s}. 
\subsection{Lax equation}
\begin{proposition}
Let $w:\mathbb R\to\mathbb C$ be a Schwartz function. 
Let $s>0$ be such that $F_w(s)\neq 0$.
Then, the unique solution
	$U(\lambda;s)$ to the RH problem for $U$ solves the following linear differential equation,
\be\label{eq:diff s eq U}
\pa_s U (\lambda;s)= M(\lambda;s) U(\lambda;s), \qquad M(\lambda;s) = \i\begin{pmatrix}\lambda & -\beta(s)\\ \gamma(s)&-\lambda
\end{pmatrix},
\ee
where \be
\beta(s)\coloneqq 2\left[U_1(s)\right]_{1,2},\quad 
\gamma(s)\coloneqq 2\left[U_1(s)\right]_{2,1},\ee
and $U_1(s)$ is defined by \eqref{eq:asy U}.
Moreover, if $w$ is even, 
\be\label{eq:sidentities}\beta(s)=-\gamma(s),\qquad \partial_s\alpha=\i \gamma^2,\quad\text{with}\quad \alpha(s)\coloneqq 2\left[U_1(s)\right]_{1,1}.\ee
\end{proposition}
\begin{proof}
Since $J_U(\lambda)$ does not depend on $s$, we can deduce that
\be 
M(\lambda;s) \coloneqq \left(\pa_s U(\lambda;s)\right) U(\lambda;s)^{-1}
\ee 
is analytic for all $\lambda \in \C$, since it has no jump on the real line and smooth boundary values. Recall the asymptotic behavior of $U$ as $\lambda\to\infty$, given in \eqref{eq:asy U}, and write 
\be\label{def:U1}
U_1(s) = \frac{1}{2}\begin{pmatrix}
	\alpha(s) &\beta(s)\\
	\gamma(s) & -\alpha(s)
\end{pmatrix}.
\ee
Here we note that $U_1(s)$ has indeed zero trace, since the determinant of $U$ is identically equal to $-1$, which can be checked from the RH conditions for $U$.

Substituting these asymptotics in the definition of $M$, and using Liouville's theorem, we conclude that $M(\lambda;s)$ is a polynomial of degree 1 in $\lambda$, more precisely we find
\be
M(\lambda;s) = \i\begin{pmatrix}\lambda & -\beta(s)\\ \gamma(s)&-\lambda
\end{pmatrix}.
\ee

If $w$ is even, we have in addition the symmetry relation
\be\label{eq: sym U}
\s_1 U(-\lambda;s)\s_3 = U(\lambda;s).
\ee
Indeed, we verify that the left hand side also satisfies the RH problem for $U$, hence the equality follows by uniqueness of the solution.
So 
\be
M(\lambda;s) = \s_1 \pa_s U(-\lambda;s) U(-\lambda;s)^{-1} \s_1 =\i \begin{pmatrix}
\lambda&\gamma(s)\\-\beta(s) &-\lambda
\end{pmatrix}
\ee 
from which we conclude that $\beta(s)=-\gamma(s)$.
Moreover, expanding the identity $\det U\equiv -1$ as $\lambda\to\infty$ yields the equation 
\be\label{eq:alphatilde}4\widetilde\alpha-\alpha^2=-\gamma^2,\ee
where $\widetilde\alpha$ is twice the $(1,1)$-entry of the $O(\lambda^{-2})$-term in \eqref{eq:asy U}. Expanding $(\partial_s U)U^{-1}$ as $\lambda\to\infty$ (in particular, the $(1,1)$-entry of the $1/\lambda$-term) implies \[2\i\partial_s\alpha=-\gamma^2 +4 \widetilde\alpha-\alpha^2.\]
Combining the last two identities, we get the second equality in \eqref{eq:sidentities}.
\end{proof}
\begin{corollary}\label{cor:equationsphipsi}
Let $w:\mathbb R\to\mathbb C$ be a Schwartz function, and let $s>0$ be such that $F_w(s)\neq 0$. Then, $\phi(\lambda;s):=U_{1,1}(\lambda;s)$ and $\psi(\lambda;s):=U_{2,1}(\lambda;s)$ satisfy the system of equations
\begin{align}
	& \label{eq:ZS1}\partial_s\phi = \i \lambda \phi -\i\beta(s)\psi,\\
	& \label{eq:ZS2}\partial_s\psi =  \i\gamma(s)\phi -\i\lambda \psi.
\end{align}
Moreover, if $w$ is even,
\be\label{eq:ZSsym}
\psi(\lambda;s) = \phi(-\lambda;s),\qquad 
\partial_s\phi(\lambda;s) = \i \lambda \phi(\lambda;s) +\i\gamma(s)\phi(-\lambda;s).
\ee
\end{corollary}

We now look for an expression of $\beta(s), \gamma(s)$ in terms of $\phi, \psi$.
\begin{proposition}
	\label{prop:integral repre gamma}
Let $w:\mathbb R\to\mathbb C$ be a Schwartz function, and let $s>0$ be such that $F_w(s)\neq 0$.
We have the relations
\begin{align}
	\label{eq:ortog}
	&\int_\R\phi(\lambda;s) \psi(\lambda;s) w'(\lambda)\d \lambda=0,\\
	\label{eq:repr beta}
	&\beta(s)= -\frac{1}{2\pi s} \int_{\R}\phi^2(\lambda;s)w'(\lambda)\d\lambda,\\
	\label{eq: repr gamma}
	& \gamma(s) = -\frac{1}{2\pi s} \int_{\R}\psi^2(\lambda;s)w'(\lambda)\d\lambda.
\end{align}
For $w$ even, they reduce to
\be
\label{eq: int repres gamma sym}
\gamma(s) = \frac{1}{2\pi s}\int_\R \phi^2(\lambda;s)w'(\lambda)\d\lambda,
\ee
and
\be 
\int_\R \phi(\lambda;s)\phi(-\lambda;s) w'(\lambda) \d\lambda=0.
\ee 
\end{proposition}
\begin{proof} 
We start by computing 
\be 
\Delta\left(\left(\pa_{\lambda}U\right) U^{-1}\right)(\lambda;s)\coloneqq \left(\left(\pa_{\lambda}U\right) U^{-1}\right)_+(\lambda;s) - \left(\left(\pa_{\lambda}U\right) U^{-1}\right)_-(\lambda;s).
\ee
Using the jump condition for $U$ and the fact that $\det U=-1$, we compute
\begin{align}
\left(\left(\pa_{\lambda}U\right) U^{-1}\right)_+(\lambda;s) &=\left(\left(\pa_{\lambda}U\right) U^{-1}\right)_-(\lambda;s) + U_-(\lambda;s)\begin{pmatrix}
	0&-w'(\lambda)\\
	0&0
\end{pmatrix}U_-^{-1}(\lambda;s) \\
&= \left(\left(\pa_{\lambda}U\right) U^{-1}\right)_- + 
\begin{pmatrix}
	-\phi (\lambda)\psi(\lambda;s) w'(\lambda)&\phi^2(\lambda;s)w'(\lambda)\\
	-\psi^2(\lambda;s)w'(\lambda)&\phi (\lambda;s)\psi(\lambda;s) w'(\lambda)
\end{pmatrix}.
\end{align}
Integrating this relation entry by entry along the real line, we get
\begin{align}
	&\int_\R \Delta \left(\left(\pa_{\lambda}U\right) U^{-1}\right)_{11} (\lambda;s)\d\lambda = -\int_\R\phi(\lambda;s) \psi(\lambda;s) w'(\lambda)\d \lambda, \\
	& \int_\R \Delta \left(\left(\pa_{\lambda}U\right) U^{-1}\right)_{12}(\lambda;s)\d\lambda= \int_{\R}\phi^2(\lambda;s)w'(\lambda)\d\lambda, \\ 
	&\int_\R \Delta \left(\left(\pa_{\lambda}U\right) U^{-1}\right)_{21}(\lambda;s)\d\lambda =-\int_{\R} \psi^2(\lambda;s)w'(\lambda)\d\lambda, \\ 
	&\int_\R \Delta \left(\left(\pa_{\lambda}U\right) U^{-1}\right)_{22} (\lambda;s)\d\lambda = \int_\R\phi(\lambda;s) \psi(\lambda;s) w'(\lambda)\d \lambda.
\end{align}
Now the integrals in the left hand side can be computed by contour deformation and residue computation:
	\be
	\int_\R \Delta \left(\left(\pa_{\lambda}U\right) U^{-1}\right)_{ij} (\lambda;s)\d\lambda =2\pi\i \mathrm{res}_{\lambda=\infty} \left[\left(\pa_{\lambda}U\right) U^{-1}\right]_{ij},
	\ee
where $\mathrm{res}_{\lambda=\infty}$ denotes the formal residue at infinity (note that the singularity at $\infty$ of $\left[\left(\pa_{\lambda}U\right) U^{-1}\right]_{ij}$ is not necessarily isolated, but that the function does admit an asymptotic Laurent series expansion at $\infty$).

Using the asymptotic condition for $U$ \eqref{eq:asy U} we get
\be
\mathrm{res}_{\lambda=\infty} 
\left(\left(\pa_{\lambda}U\right) U^{-1}\right) =  -\i s \left[U_1(s), \s_3\right] = -\i s
\begin{pmatrix}
	0 &-\beta(s)\\
	\gamma(s)&0
\end{pmatrix}.
\ee
Thus we obtain formulae \eqref{eq:ortog}, \eqref{eq:repr beta}, \eqref{eq: repr gamma}.

Finally, if $w$ is even then as already noticed we have $\beta(s)=-\gamma(s)$ and $\psi(\lambda;s)=\phi(-\lambda;s)$. Thus in this case the same computation yields to
\begin{align}
	&\int_\R\phi(\lambda;s) \phi(-\lambda;s) w'(\lambda)\d \lambda=0\\
	& \gamma(s) = \frac{1}{2\pi s} \int_{\R}\phi^2(\lambda;s)w'(\lambda)\d\lambda.
\end{align}
\end{proof}
Substituting the integral identities from Proposition \ref{prop:integral repre gamma} into the differential equations from Corollary \ref{cor:equationsphipsi}, we recognize the integro-differential equations for $\phi,\psi$ appearing in Theorem \ref{thm:s}. The $\lambda\to\infty$ asymptotics for $\phi,\psi$ present in Theorem \ref{thm:s} follow directly from the RH conditions for $U$.

\subsection{Second logarithmic derivative}

In order to complete the proof of Theorem \ref{thm:s}, it remains to prove the expression for the second logarithmic $s$-derivative of $F_w$ in terms of $\phi$ and $\psi$ {(notice that the closed expression we have is actually for the derivative of $s$ times the logarithmic derivative of $F_w$, slightly different from the actual second logarithmic derivative as in the Airy case, comparing to equation \eqref{eq:airy fred det})}.
\begin{proposition}
Let $w:\mathbb R\to\mathbb C$ be a Schwartz function, and let $s>0$ be such that $F_w(s)\neq 0$.
Then we have the identity
\be
	\label{eq:Fw second log der}
	\partial_s s \partial_s\log F_w(s)= \frac{1}{\pi} \int_{\R}  \phi (\lambda;s)\psi(\lambda;s) w'(\lambda)\lambda\d \lambda,
	\ee
which, in case $w$ is even, reduces to	 
	\be \label{eq:Fw second log der even}
	\partial_s s\partial_s\log F_w(s)= \frac{1}{\pi} \int_{\R} \phi (\lambda;s)\phi(-\lambda;s)  w'(\lambda)\lambda\d \lambda. 
\ee
\end{proposition}
\begin{proof}
We first observe that using the notation in Corollary \ref{cor:equationsphipsi}, we can express
	\be
\left[ U_+^{-1}(\lambda;s)\partial_\lambda U_+(\lambda;s)\right]_{2,1} = \psi(\lambda;s) \pa_\lambda\phi(\lambda;s) -\phi(\lambda,s) \pa_\lambda\psi(\lambda;s).
\ee	
	Thus by Proposition \ref{prop:detU}, we get
\be
\pa_s \log  F_w(s) =\frac{1}{2\pi\i s}\int_\R \lambda w'(\lambda)(\psi(\lambda;s) \pa_\lambda\phi(\lambda;s) -\phi(\lambda;s) \pa_\lambda\psi(\lambda;s))\d \lambda.
\ee
	Taking another $s$-derivative we get 
\be
\pa_s s\pa_s\log F_w(s) = \frac{1}{2\pi\i} \int_\R \lambda w'(\lambda) \pa_s (\psi(\lambda;s) \pa_\lambda\phi(\lambda;s) -\phi(\lambda;s) \pa_\lambda\psi(\lambda;s))\d \lambda.
\ee

	The integrand can be further simplified by using the differential equations in $s$ for $\phi, \psi$  \eqref{eq:ZS1}, \eqref{eq:ZS2}. In particular we obtain 
	\be
	\pa_s (\psi(\lambda;s) \pa_\lambda\phi(\lambda;s) -\phi(\lambda;s)\pa_\lambda\psi(\lambda;s))=2\i \phi(\lambda;s)\psi(\lambda;s).
	\ee
	Thus we can conclude
	\be
\pa_s s\pa_s\log F_w(s)= \frac{1}{\pi} \int_\R \lambda w'(\lambda) \phi (\lambda;s)\psi(\lambda;s) \d \lambda.
	\ee

	In the case where $w$ is even, then $\psi(\lambda;s)=\phi(-\lambda;s)$ and equation \eqref{eq:Fw second log der even} is obtained as well. 
\end{proof}
The proof of Theorem \ref{thm:s} is now completed.

\section{The corresponding PDE}
\label{section:ys}
In the previous section, we considered the dependence on the parameter $s$ of  $F_w(s)$ and $Q_W(y,s)$, and of the associated RH solution $U$. 
In this section we will study the $y$-dependence of $Q_W(y,s)$. 

\subsection{Lax pair}
We now consider the RH problem for $U$ corresponding to $w(\lambda)=W(\lambda^2-y)$, with $W:\mathbb R\to\mathbb C$ such that $w$ is a Schwartz function for every $y\in\mathbb R$. For the ease of notations, we write
\[U(\lambda;y,s):=U^{(w)}(\lambda;s),\]
and accordingly for 
$U_1(y,s)$.
  Since the jump condition for $U$ does not depend on $s$, the differential equation \eqref{eq:diff s eq U} for $U(\lambda;s)$ in $s$ derived in the previous section continues to hold for $U(\lambda;y,s)$. 
Instead of writing $U_1=U_1(y,s)$ as in \eqref{def:U1}, it is now more convenient to write
\be\label{def:U12}
U_1(y,s)=-\i \begin{pmatrix}p(y,s)&-q(y,s)\\
q(y,s)&-p(y,s).\end{pmatrix}
\ee
In other words, we set $p=\i\alpha/2$, $q=\i\gamma=-\i\beta/2$, which we can do because $w$ is even, such that $\beta=-\gamma$ by \eqref{eq:sidentities}.  
 The differential equation \eqref{eq:diff s eq U}
  then reads
  \be\label{eq:Laxs2}
\partial_s U(\lambda;y,s)=\begin{pmatrix}\i\lambda & -\i q(y,s)\\ \i q(y,s)&-\i\lambda
\end{pmatrix}U(\lambda;y,s),
  \ee
  and the second identity in \eqref{eq:sidentities} becomes
  \be\label{eq:idpsq}\partial_s p(y,s)=q(y,s)^2.\ee

Observe that the jump condition for $U(\lambda;y,s)$ does depend on $y$, and 
because of this, we cannot directly derive a Lax equation in $y$, in order to obtain the PDE \eqref{PDE} as compatibility relation of a Lax pair.
Instead, inspired by \cite[Section 5]{IIKS}, we introduce the differential operator 
\be 
\label{eq: D lambda beta}
D_{\lambda,y} = \pa_\lambda+2\lambda \pa_y,
\ee
which allows us to obtain a second linear differential equation for $U(\lambda;y,s)$.

\begin{proposition} 
Let $y\in\mathbb R$ and let $W:\mathbb R\to\mathbb C$ be such that $W(.^2-y)$ is a Schwartz function.
Then, for any $s>0$ such that $Q_W(y,s)\neq 0$,
we have
	\be
	\label{eq: diff eq U Dlambdabeta}
D_{\lambda,y} U(\lambda;y,s) = L(\lambda;y,s) U(\lambda;y,s), \quad L(\lambda;y,s) = \begin{pmatrix}
\i s -\i\pa_y p(y,s) &\i\pa_y q(y,s)\\ -\i\pa_y q(y,s)& -\i s +\i\pa_y p(y,s) 
\end{pmatrix}.
\ee 
\end{proposition}
\begin{proof}
Recall that $U$ is differentiable with respect to $y$, and that \eqref{eq:asy U} can be differentiated with respect to $y$.

Notice that, since $w(\lambda)=W(\lambda^2-y)$, we have
\be 
D_{\lambda,y} (w(\lambda)) =2\lambda W'(\lambda^2-y) - 2\lambda  W'(\lambda^2-y)=0.
\ee 
Thus, defining 
\be 
L(\lambda;y,s)= D_{\lambda,y} U(\lambda;y,s) U(\lambda;y,s)^{-1}, 
\ee 
we can conclude, from the first two properties of the RH problem for $U$ and the fact that $L$ has continuous boundary values,  that $L(\lambda;y,s)$ is an entire function of $\lambda$. Moreover, from the asymptotic expansion of $U(\lambda;y,s)$ as $\lambda\to\infty$ together with Liouville's theorem, we conclude that $L(\lambda;y,s)$  is independent of $\lambda$ and that it has precisely the form given in \eqref{eq: diff eq U Dlambdabeta}.
\end{proof}
The compatibility condition between equations \eqref{eq:diff s eq U} and \eqref{eq: diff eq U Dlambdabeta}, obtained by imposing that $\partial_s D_{\lambda,y} U=D_{\lambda,y}\partial_s U$, reads
\be \label{eq:compatibilityPDE}
\pa_s L-D_{\lambda,y}M+\left[L,M\right]=0,
\ee 
and
is equivalent to a coupled system of PDEs for $p(y,s) $ and $q(y,s)$, namely
\begin{align}
&\label{eq:PDEgamma1} \partial_y\partial_sp(y,s)=2 q(y,s) \partial_yq(y,s),\\
& \label{eq:PDEgamma2}\partial_y\partial_sq(y,s)+2q(y,s)(s -\partial_y p(y,s))=0.
\end{align}
The first identity reveals nothing new, since it is the $y$-derivative of \eqref{eq:idpsq}.
Dividing the second identity by $2q(y,s)$, taking another $s$-derivative, and using again \eqref{eq:idpsq}, we obtain a PDE for $q(y,s)$ only, which is precisely \eqref{PDE}.

\subsection{Local differential identity}
The second logarithmic $s$-derivative of $Q_W\left(y,s\right)$ can also be expressed in terms of the solution $q(y,s)$ of the PDE \eqref{PDE}.

\begin{proposition}\label{prop:eqQgamma}
Let $W:\mathbb R\to\mathbb C$ be such that $W(.^2-y)$ is a Schwartz function for every $y\in\mathbb R$. Let $s>0$, $y\in\mathbb R$ be such that $Q_W\left(y,s\right)\neq 0$, and let $q(y,s)$ be as in \eqref{def:U12}.
Then, we have
	\be  \label{eq:DDq}
	\pa_s^2 \log Q_W\left(y,s\right) = -q(y,s)^2,
	\ee
	and $q$ solves the PDE \eqref{PDE}.
\end{proposition}
\begin{proof}
We adapt formula \eqref{eq:Fw second log der even} to the case of $Q_W(y,s)=F_{w}(s)$, so that
	\be
	\pa_s s\pa_s\log Q_W(y,s) = \frac{1}{\pi} \int_\R \lambda w'(\lambda) \phi (\lambda;y,s)\phi(-\lambda;y,s) \d \lambda.
	\ee
We can now repeat the same procedure used in the proof of Proposition \ref{prop:integral repre gamma} to see that 
\be
\Delta\left(\lambda \pa_\lambda UU^{-1}\right)= 
\begin{pmatrix}	-\lambda\phi (\lambda;y,s)\phi(-\lambda;y,s) w'(\lambda)&\lambda\phi^2(\lambda;y,s)w'(\lambda)\\
	-\lambda\phi^2(-\lambda;y,s)w'(\lambda)&\lambda\phi (\lambda;y,s)\phi(-\lambda;y,s) w'(\lambda)
\end{pmatrix}.
\ee 
We consider the $(1,1)$ entry and we take the integral over $\R$:
\be 
-\int_\R	\lambda\phi (\lambda;y,s)\phi(-\lambda;y,s) w'(\lambda)\d\lambda = \int_\R \left[ \Delta\left(\lambda \pa_\lambda UU^{-1}\right)\right]_{1,1}\d\lambda.
\ee 
As we did in the proof of Proposition \ref{prop:integral repre gamma}, we can now  compute the integral in the right hand side by residue computation: expressing \eqref{eq:alphatilde} in terms of $p(y,s)$ and $q(y,s)$ and using \eqref{eq:idpsq}, we obtain
\begin{align}\nonumber
-\frac{1}{\pi}\int_\R \left[ \Delta\left(\lambda \pa_\lambda UU^{-1}\right)\right]_{1,1}\d\lambda& = -2\i\mathrm{res}_{\lambda=\infty} \left[ \left(\lambda \pa_\lambda UU^{-1}\right)\right]_{1,1}=-p(y,s)-sq(y,s)^2
\\
&=-p(y,s)-s\partial_sp(y,s).
\end{align}
We conclude that 
\be
\pa_s s\pa_s\log Q_W(y,s)= -\partial_s(sp(y,s)).
\ee 
We can now integrate this equation in $s$, between $0$ and $s$. In Proposition \ref{prop:sto0} and Proposition \ref{prop:RHsto0}, we showed respectively that $s\partial_s\log Q_W(y,s)$ and $sp(y,s)$ tend to zero as $s\rightarrow 0$. Hence the integration constant can be set to zero, and we obtain
\be\label{id:alphaQ}
\pa_s\log Q_W(y,s) =-p(y,s).
\ee
By differentiating in $s$ and replacing again equation \eqref{eq:sidentities}, we conclude that
\be 
\pa_s^2\log Q_W(y,s) = -q^2(y,s),
\ee
and obtain \eqref{eq:DDq}. Since we already showed that $q$ solves \eqref{PDE}, the result is proved. 
\end{proof}

\subsection{Sigma-form of the PDE}
{We are now going to derive the PDE written in \eqref{PDEsigma} for $\sigma_W(y,s)=\log Q_W(y,s)$. We call it a \textit{sigma-form} referring to the fact that in the classical case (when $w=\chi_{(-1/2,1/2)}$), the quantity $s\partial_s\log F(s;1)$ evolves according to the Painlevé V sigma-form, recall equations \eqref{eq:PVsigmaform} and \eqref{eq:nu and fred det zero t}.}
\begin{proposition}\label{prop:eqQsigma}
Let $y\in\mathbb R$ and let $W:\mathbb R\to\mathbb C$ be such that $W(.^2-y)$ is a Schwartz function. Let
$\sigma(y,s)=\sigma_W(y,s)=\log Q_W(y,s)$. Then, for any $s>0$ such that $Q_W\left(y,s\right)\neq 0$,
$\sigma$ solves the PDE \eqref{PDEsigma}.
\end{proposition}
\begin{proof}
Recall from Proposition \ref{prop:eqQgamma} and \eqref{id:alphaQ} that
\be
q^2(y,s)=-\partial_s^2\sigma(y,s)=\partial_s p(y,s).
\ee
We can thus re-write 
\eqref{eq:PDEgamma2} 
as
\be\label{eq:systempq}
\partial_{s}\partial_y q=-2sq+2q \partial_y p.
\ee
Multiplying with $\partial_yq$, we find
\[\frac{1}{2}\partial_s\left((\partial_y q)^2\right)=-s\partial_y(q^2)+\partial_yp\ \partial_y (q^2),\]
or
\[\frac{1}{2}\partial_s\left((\partial_y q)^2\right)=s\partial_s^2\partial_y\sigma+{\partial_{s}\partial_y\sigma\ \partial_s^2\partial_y\sigma}.\]
Integrating in $s$, we get
\[\frac{1}{2}(\partial_y q)^2=s\partial_{s}\partial_y\sigma - \partial_{y}\sigma +\frac{1}{2}\left(\partial_{s}\partial_y\sigma\right)^2+c,\]
and the integration constant $c$ is equal to $0$. Indeed it follows from Propositions \ref{prop:sto0} and  \ref{prop:RHsto0} that 
\[\lim_{s\to 0}\partial_y p(y,s)=\lim_{s\to 0}\partial_y q(y,s)=\lim_{s\to 0}\partial_y \sigma(y,s)=0,\]
hence by taking the limit $s\to 0$ on both sides we obtain $c=0$. 
We also have \[\frac{(\partial_y q)^2}{2}=\frac{\left(\partial_y(q^2)\right)^2}{8q^2}=-\frac{\left(\partial_s^2\partial_y\sigma\right)^2}{8\partial_{s}^2\sigma}.\]
Substituting this, we obtain
\[\frac{\left(\partial_s^2\partial_y\sigma\right)^2}{4\partial_{s}^2\sigma}=-2s\partial_{s}\partial_y\sigma + 2\partial_{y}\sigma -\left(\partial_{s}\partial_y\sigma\right)^2,\]
which we recognize as \eqref{PDEsigma}.
\end{proof}
Propositions \ref{prop:eqQgamma} and \ref{prop:eqQsigma} together yield Theorem \ref{thm:q}. Combining this with \eqref{asTrace}, we also obtain Theorem \ref{thm:assigma}.

\section{Proof of Theorem \ref{thm:scattering}}\label{section:as}

We are now ready to prove Theorem \ref{thm:scattering}. Let $f$ satisfy the assumptions from Theorem \ref{thm:scattering}, and define
\be \label{def:WF}W(r)=-2\int_0^{\infty}f'(-u^2-r)\d u.\ee
We then have
\be
W'(r)=2\int_0^{\infty}f''(-u^2-r)\d u.
\ee
Since $f$ is $C^\infty$ and $f$ as well as its derivatives decay fast at $-\infty$, it is easy to see that $W$ is infinitely many times differentiable, and 
that $\lim_{r\to +\infty}r^kW^{(\ell)}(r)=0$ for every positive integers $k,\ell$: indeed, we have
\[|W^{(\ell)}(r)|\leq 2\int_0^\infty |f^{(\ell+1)}(-u^2-r)|\d u=\int_r^\infty |f^{(\ell+1)}(-t)|\frac{\d t}{\sqrt{t-r}},\]
and by the fast decay of $f^{(\ell+1)}$ at $-\infty$, the right hand side decays fast as $r\to +\infty$.

 As a consequence, $w(u)=W(u^2-y)$ is a Schwartz function, and it is an admissible choice to apply Theorem \ref{thm:assigma}.
Hence,
we have
\[\lim_{s\to 0}\frac{1}{s}\sigma_W(y,s)=-\frac{2}{\pi}\int_0^{+\infty}W(\lambda^2-y)\d\lambda.\]
Now we substitute \eqref{def:WF}, and obtain
\[\lim_{s\to 0}\frac{1}{s}\sigma_W(y,s)=\frac{4}{\pi}\int_0^{\infty}\left(\int_0^{+\infty}f'(-u^2-\lambda^2+y)\d u\right)\d\lambda.\]
Changing to polar coordinates, we get
\[\lim_{s\to 0}\frac{1}{s}\sigma_W(y,s)=2\int_0^{\infty} f'(-\rho^2+y)\rho \d \rho =\int_0^{\infty}f'(-v+y)\d v=f(y),\]
since $f(-\infty)=0$.

\section*{Acknowledgments} 
The authors were supported by Fonds de la Recherche Scientifique-FNRS under EOS project O013018F, and by CNRS International Research Network PIICQ. TC was also supported by FNRS Research Project T.0028.23. The authors are grateful to Thomas Bothner for drawing their attention to the PDE \eqref{PDEsigma} in \cite{IIKS}, and to Alexander Its, Peter Miller and Gregory Schehr for useful comments. 

{
	}
\end{document}